\documentclass[12pt,reqno]{article}
\usepackage{graphicx}
\usepackage{color}
\usepackage{amsmath}
\usepackage{amssymb}
\usepackage{amscd}
\usepackage{amsthm}
\usepackage{amsopn}
\usepackage{xspace}
\usepackage{verbatim}
\usepackage{epic}
\usepackage{eepic}
\usepackage{stmaryrd}
\usepackage{empheq}
\usepackage[active]{srcltx}
\usepackage[margin=3cm]{geometry}
\definecolor{red}{rgb}{1,0,0}
\definecolor{green}{rgb}{0,1,0}
\definecolor{SeaGreen}{RGB}{46,139,87}
\definecolor{Maroon}{RGB}{128,0,0}

\newcommand{\C}{{\mathbb{C}}}
\newcommand{\R}{{\mathbb{R}}}

\newcommand{\A}{\mathcal A}
\newcommand{\B}{\mathcal B}
\newcommand{\CC}{\mathcal C}

\makeatletter\newcommand*{\at}{@}\makeatother
\def\Dg {{\mathcal D}}

\def\Hg {{\mathcal H}}

\def\Jg {{\mathcal J}}

\newcommand{\OO}{\mathcal O}

\def\Sg {{\mathcal S}}
\def\Ug {{\mathcal U}}

\def\curl{\text{\rm curl}}

\def\curl{\text{\rm curl\,}}

\def\Div{\text{\rm div\,}}

\def\loc{\text{\rm loc}}

\renewcommand{\Im}{{\rm Im\,}}
\newcommand {\pa}{\partial}

\def\O{\Omega}

\def\q{\quad}
\def\qq{\qquad}

\def\k{\kappa}

\def\0{\mathbf  0}

\newcommand{\eq}{\begin{equation}}
\newcommand{\eeq}{\end{equation}}

\def\Xint#1{\mathchoice
{\XXint\displaystyle\textstyle{#1}}%
{\XXint\textstyle\scriptstyle{#1}}%
{\XXint\scriptstyle\scriptscriptstyle{#1}}%
{\XXint\scriptscriptstyle\scriptscriptstyle{#1}}%
\!\int}
\def\XXint#1#2#3{{\setbox0=\hbox{$#1{#2#3}{\int}$ }
\vcenter{\hbox{$#2#3$ }}\kern-.6\wd0}}

\def\dashint{\Xint-}

 

\errorcontextlines=0 \numberwithin{equation}{section}
\theoremstyle{plain}
\newtheorem{theorem}{Theorem}[section]
\newtheorem{lemma}[theorem]{Lemma}
\newtheorem{proposition}[theorem]{Proposition}

\newtheorem{remark}[theorem]{Remark}
\newtheorem{corollary}[theorem]{Corollary}

\newtheorem{conjecture}[theorem]{Conjecture}

\begin{document}
\date{}
\title{Mixed normal-superconducting states in the presence of strong
  electric currents}
\author{Yaniv Almog \thanks{Department
    of Mathematics, Louisiana State University,
    Baton Rouge, LA 70803, USA. \newline almog \at math.lsu.edu}, Bernard Helffer
   \thanks{Laboratoire de Math\'ematiques, Universit\'e Paris-Sud 11 and
     CNRS, B\^at 425, 91 405 Orsay Cedex, France
 and Laboratoire de Math\'ematiques Jean Leray, Universit\'e de
 Nantes.  \newline Bernard.Helffer\at math.u-psud.fr},
 and Xing-Bin Pan \thanks{Department of Mathematics,
 East China Normal University, and NYU-ECNU Institute of Mathematical
 Sciences at NYU Shanghai, Shanghai 200062, P.R. China. xbpan\at math.ecnu.edu.cn}}

\maketitle
\large
\bibliographystyle{siam}

\begin{abstract}
  We study the Ginzburg-Landau equations in the presence of large
  electric currents, that are smaller than the critical current where the
  normal state losses its stability. For steady-state solutions in the
  large $\kappa$ limit, we prove that the superconductivity order parameter is
  exponentially small in a significant part of the domain, and small
  in the rest of it. Similar results are obtained for the time-dependent problem, in
  continuation of the paper by the two first authors \cite{alhe14}. We
  conclude by obtaining some weaker results, albeit similar, for
  steady-state solutions in the large domain limit.
\end{abstract}

\tableofcontents

\newpage
\section{Introduction}\label{Section1}
Consider the time-dependent Ginzburg-Landau system of equations
\begin{subequations}
\label{eq:1}
\begin{alignat}{2}
   \frac{\partial\psi}{\partial t}-  \nabla_{\kappa A}^2\psi +i\kappa\phi \psi=\kappa^2(1-|\psi|^2)\psi & \quad \text{ in }(0,+\infty)\times  \Omega \,,\\
 \frac{1}{c}\Big(\frac{\partial A}{\partial t}+\nabla\phi\Big)+ \curl^2A
  =\frac{1}{\kappa}\Im(\bar\psi\nabla_{\kappa A}\psi) &
  \quad \text{ in } (0,+\infty)\times   \Omega\,, \\
  \psi=0 &\quad \text{ on } (0,+\infty)\times  \partial\Omega_c\,, \\
  \nabla_{\kappa A}\psi\cdot\nu=0 & \quad \text{ on } (0,+\infty)\times  \partial\Omega_i\,, \\
  \frac{\partial\phi}{\partial\nu} = - c\kappa J(x) &\quad \text{ on } (0,+\infty)\times \partial\Omega_c\,, \\
  \frac{\partial\phi}{\partial\nu}=0  &\quad \text{ on } (0,+\infty)\times  \partial\Omega_i \,, \\[1.2ex]
  \dashint_{\partial\Omega}\curl A(t,x)  \, ds = \kappa h_{ex} \,,&\quad {\text{ on } (0,+\infty)}\\
  \psi(0,x)=\psi_0(x)  & \quad \text{ in } \Omega\,, \\
A(0,x)=A_0(x) & \quad \text{ in } \Omega \,.
\end{alignat}
\end{subequations}
In the above $\psi$ denotes the order parameter, $A$  the
magnetic potential, $\phi$  the electric potential, $\kappa$
the Ginzburg-Landau parameter, which is a material property, and
$$c=\kappa^2/\sigma$$ where $\sigma$ is the normal conductivity of the
sample. Finally, $h_{ex}$, or the average magnetic field on $\partial\Omega$
divided by $\kappa$, is constant in time. Length has been scaled with
respect to the penetration depth (see \cite{alhe14}). \\
 Unless otherwise stated we
shall assume in the sequel that $$\kappa\geq1\,.$$ We further assume that
$(\psi_0,A_0)\in H^1(\Omega,\C)\times H^1(\Omega,\R^2)$ and that
\begin{equation}\label{condinit}
  \|\psi_0\|_\infty\leq 1 \,.
\end{equation}
We use the notation $\nabla_A=\nabla-iA$, $\Delta_A=\nabla_A^2$ and $ds$ for the induced
measure on $\pa \Omega$.  We have also used above the standard notation
\begin{displaymath}
  \dashint_{\partial \Omega}  = \frac{1}{\ell(\partial \Omega)}\int \,.
\end{displaymath}
The domain $\Omega\subset\subset\R^2$ has the same characteristics as in \cite{alhe14},
in particular its boundary $\pa \Omega$ contains a smooth interface,
denoted by $\partial\Omega_c\,$, with a conducting metal which is at the normal
state. Thus, we require that $\psi$ vanishes on $\partial\Omega_c$ in (\ref{eq:1}c).

We make the following assumptions on the current $J$
\begin{equation}
\label{eq:2}
(J1)\quad J\in C^2(\overline{\partial\Omega_c},\R)\,,
\end{equation}
\begin{equation}
  \label{eq:3}
(J2) \quad \int_{\partial\Omega_c}J \,ds=0 \,,
\end{equation}
and
\begin{equation}\label{eq:75aa}
(J3)\quad  \mbox{the sign of } J \mbox{ is constant on each
connected component of } \partial\Omega_c\,.
\end{equation}
We allow for $J\neq0$ at the corners despite the fact that no current is
allowed to enter the sample through the insulator.  The rest of
$\partial\Omega$, denoted by $\partial\Omega_i$ is adjacent to an insulator. By
convention, we extend $J$ as equal to $0$ on $\partial \Omega_i$.

To simplify some of our regularity arguments we introduce the
following geometrical assumption (for further discussion we refer the
reader to Appendix A in \cite{alhe14}) on $\partial\Omega$:
\begin{equation}\label{propertyR}
(R1)\,\left\{
\begin{array}{l}
(a) \; \pa \Omega_i \mbox{ and } \pa \Omega_c  \mbox{ are of class } C^3\,;\\
(b) \mbox{ Near each edge, } \pa
\Omega_i \mbox{  and  } \pa \Omega_c
 \mbox{ are} \\ \quad  \mbox{ flat and meet with  an angle of } \frac \pi 2\,.
\end{array}\right.
\end{equation}
  We also require that
\begin{equation}\label{hyptopolo}
(R2) \quad\quad \mbox{Both }  \partial\Omega_c \mbox{  and } \partial\Omega_i \mbox{  have two components}.
\end{equation}
\begin{figure}
  \centering
\setlength{\unitlength}{0.0005in}
\begingroup\makeatletter\ifx\SetFigFont\undefined%
\gdef\SetFigFont#1#2#3#4#5{%
  \reset@font\fontsize{#1}{#2pt}%
  \fontfamily{#3}\fontseries{#4}\fontshape{#5}%
  \selectfont}%
\fi\endgroup%
{\renewcommand{\dashlinestretch}{30}\begin{picture}(2766,6795)(0,-10)
\path(1275,300)(2700,300)
\path(375,6450)(1875,6450)
\path(1575,300)(1575,1125)
\blacken\path(1605.000,1005.000)(1575.000,1125.000)(1545.000,1005.000)(1605.000,1005.000)
\path(1800,300)(1800,1200)
\blacken\path(1830.000,1080.000)(1800.000,1200.000)(1770.000,1080.000)(1830.000,1080.000)
\path(2025,300)(2025,900)
\blacken\path(2055.000,780.000)(2025.000,900.000)(1995.000,780.000)(2055.000,780.000)
\path(2325,300)(2325,675)
\blacken\path(2355.000,555.000)(2325.000,675.000)(2295.000,555.000)(2355.000,555.000)
\path(1350,300)(1425,1050)
\blacken\path(1442.911,927.610)(1425.000,1050.000)(1383.208,933.581)(1442.911,927.610)
\path(2550,300)(2475,750)
\blacken\path(2524.320,636.565)(2475.000,750.000)(2465.136,626.701)(2524.320,636.565)
\path(525,5775)(525,6450)
\blacken\path(555.000,6330.000)(525.000,6450.000)(495.000,6330.000)(555.000,6330.000)
\path(900,5775)(900,6450)
\blacken\path(930.000,6330.000)(900.000,6450.000)(870.000,6330.000)(930.000,6330.000)
\path(1275,5775)(1275,6450)
\blacken\path(1305.000,6330.000)(1275.000,6450.000)(1245.000,6330.000)(1305.000,6330.000)
\path(1575,5775)(1575,6450)
\blacken\path(1605.000,6330.000)(1575.000,6450.000)(1545.000,6330.000)(1605.000,6330.000)
\path(1275,300)(1275,301)(1275,304)
	(1275,312)(1275,326)(1275,346)
	(1274,371)(1274,400)(1273,431)
	(1273,461)(1272,491)(1271,518)
	(1271,544)(1270,567)(1269,589)
	(1267,609)(1266,627)(1264,645)
	(1263,662)(1260,680)(1258,697)
	(1256,714)(1253,732)(1249,751)
	(1246,770)(1242,789)(1237,810)
	(1232,830)(1227,851)(1221,872)
	(1216,894)(1209,915)(1203,937)
	(1196,959)(1190,980)(1182,1002)
	(1175,1025)(1169,1042)(1163,1060)
	(1157,1079)(1150,1099)(1143,1119)
	(1135,1140)(1127,1162)(1119,1185)
	(1110,1208)(1101,1232)(1092,1257)
	(1082,1282)(1073,1307)(1063,1333)
	(1053,1358)(1043,1383)(1032,1409)
	(1022,1433)(1012,1458)(1002,1482)
	(992,1506)(982,1529)(972,1552)
	(963,1575)(952,1598)(942,1621)
	(932,1644)(921,1668)(910,1692)
	(899,1717)(887,1742)(876,1768)
	(864,1793)(852,1819)(840,1845)
	(828,1872)(817,1897)(805,1923)
	(794,1948)(783,1973)(773,1998)
	(763,2021)(753,2044)(744,2066)
	(736,2088)(728,2109)(720,2130)
	(713,2150)(705,2172)(698,2194)
	(691,2215)(684,2237)(677,2260)
	(671,2283)(665,2306)(659,2329)
	(653,2353)(648,2377)(643,2402)
	(638,2426)(633,2451)(629,2475)
	(624,2499)(620,2523)(617,2546)
	(613,2570)(610,2593)(606,2616)
	(603,2639)(600,2662)(597,2684)
	(594,2706)(591,2729)(588,2753)
	(585,2777)(582,2802)(579,2828)
	(576,2855)(573,2882)(569,2910)
	(566,2939)(562,2968)(559,2997)
	(556,3026)(552,3055)(549,3084)
	(546,3112)(543,3140)(540,3168)
	(537,3195)(534,3222)(531,3248)
	(528,3274)(525,3300)(522,3324)
	(520,3348)(517,3372)(514,3397)
	(511,3422)(509,3447)(506,3473)
	(503,3500)(500,3527)(497,3554)
	(494,3582)(491,3610)(487,3638)
	(484,3665)(481,3693)(478,3721)
	(475,3748)(472,3775)(469,3802)
	(466,3828)(464,3853)(461,3878)
	(458,3903)(455,3927)(453,3951)
	(450,3975)(447,3999)(445,4023)
	(442,4047)(439,4072)(437,4097)
	(434,4123)(431,4149)(428,4176)
	(425,4204)(423,4232)(420,4260)
	(417,4289)(414,4318)(412,4347)
	(409,4376)(406,4405)(404,4435)
	(402,4464)(400,4492)(397,4521)
	(395,4549)(394,4577)(392,4605)
	(390,4632)(389,4660)(388,4687)
	(386,4713)(385,4740)(384,4767)
	(383,4794)(382,4822)(381,4851)
	(381,4880)(380,4910)(379,4941)
	(379,4972)(378,5004)(378,5037)
	(377,5069)(377,5102)(376,5134)
	(376,5167)(376,5200)(376,5232)
	(375,5264)(375,5295)(375,5326)
	(375,5356)(375,5385)(375,5414)
	(375,5443)(375,5470)(375,5498)
	(375,5525)(375,5556)(375,5588)
	(375,5619)(375,5651)(375,5683)
	(375,5715)(375,5747)(375,5780)
	(375,5812)(375,5844)(375,5876)
	(375,5908)(375,5938)(375,5968)
	(375,5997)(375,6025)(375,6051)
	(375,6076)(375,6100)(375,6123)
	(375,6144)(375,6164)(375,6182)
	(375,6200)(375,6228)(375,6254)
	(375,6278)(375,6300)(375,6323)
	(375,6345)(375,6367)(375,6389)
	(375,6409)(375,6426)(375,6439)
	(375,6446)(375,6450)
\path(2704,291)(2704,293)(2704,299)
	(2705,309)(2706,324)(2707,346)
	(2709,374)(2711,409)(2713,449)
	(2715,495)(2718,545)(2721,599)
	(2724,654)(2727,711)(2730,767)
	(2732,823)(2735,877)(2738,930)
	(2740,980)(2742,1028)(2744,1073)
	(2746,1115)(2748,1155)(2749,1193)
	(2750,1228)(2751,1261)(2752,1293)
	(2753,1322)(2753,1350)(2754,1377)
	(2754,1403)(2754,1428)(2754,1464)
	(2753,1498)(2753,1531)(2752,1563)
	(2750,1595)(2749,1626)(2747,1656)
	(2745,1685)(2742,1714)(2739,1742)
	(2736,1768)(2733,1794)(2729,1819)
	(2726,1842)(2722,1864)(2718,1885)
	(2714,1905)(2709,1924)(2705,1942)
	(2701,1959)(2696,1975)(2692,1991)
	(2686,2010)(2680,2029)(2673,2048)
	(2666,2067)(2659,2086)(2651,2106)
	(2643,2126)(2634,2147)(2626,2168)
	(2616,2189)(2607,2209)(2598,2230)
	(2588,2251)(2579,2272)(2570,2292)
	(2560,2312)(2551,2333)(2542,2353)
	(2534,2371)(2526,2389)(2517,2408)
	(2508,2427)(2500,2447)(2490,2468)
	(2481,2490)(2471,2513)(2461,2536)
	(2450,2560)(2440,2584)(2429,2608)
	(2419,2633)(2408,2657)(2398,2681)
	(2387,2705)(2377,2729)(2367,2752)
	(2357,2775)(2348,2797)(2338,2819)
	(2329,2841)(2320,2861)(2312,2881)
	(2303,2902)(2294,2923)(2285,2944)
	(2276,2966)(2266,2989)(2257,3011)
	(2247,3035)(2238,3059)(2228,3083)
	(2218,3107)(2209,3131)(2199,3156)
	(2190,3180)(2181,3205)(2172,3229)
	(2163,3253)(2155,3276)(2147,3299)
	(2139,3323)(2131,3345)(2124,3368)
	(2117,3391)(2109,3414)(2102,3438)
	(2095,3461)(2088,3486)(2081,3511)
	(2074,3537)(2068,3563)(2061,3590)
	(2054,3617)(2047,3645)(2041,3672)
	(2034,3700)(2028,3728)(2022,3755)
	(2016,3782)(2011,3808)(2006,3834)
	(2001,3858)(1997,3883)(1993,3906)
	(1989,3928)(1985,3950)(1982,3971)
	(1979,3991)(1976,4011)(1974,4030)
	(1971,4050)(1969,4070)(1967,4090)
	(1965,4110)(1963,4131)(1962,4152)
	(1960,4174)(1958,4196)(1957,4219)
	(1956,4242)(1954,4266)(1953,4291)
	(1952,4316)(1951,4341)(1950,4367)
	(1949,4394)(1948,4421)(1947,4448)
	(1945,4477)(1944,4505)(1943,4535)
	(1942,4566)(1940,4592)(1939,4618)
	(1938,4646)(1936,4675)(1935,4705)
	(1933,4736)(1932,4768)(1930,4801)
	(1928,4835)(1926,4870)(1924,4906)
	(1922,4942)(1920,4979)(1918,5016)
	(1916,5054)(1915,5091)(1913,5128)
	(1911,5165)(1909,5201)(1907,5237)
	(1905,5272)(1903,5306)(1901,5339)
	(1900,5371)(1898,5402)(1897,5432)
	(1895,5461)(1894,5489)(1893,5515)
	(1892,5541)(1890,5575)(1889,5607)
	(1887,5638)(1886,5668)(1885,5698)
	(1884,5726)(1883,5754)(1883,5781)
	(1882,5808)(1881,5833)(1881,5858)
	(1880,5881)(1880,5904)(1880,5925)
	(1879,5946)(1879,5965)(1879,5984)
	(1879,6001)(1879,6018)(1879,6034)
	(1879,6050)(1879,6066)(1879,6085)
	(1879,6105)(1879,6125)(1879,6146)
	(1879,6168)(1879,6193)(1879,6219)
	(1879,6247)(1879,6278)(1879,6309)
	(1879,6340)(1879,6369)(1879,6394)
	(1879,6415)(1879,6429)(1879,6437)
	(1879,6440)(1879,6441)
\put(1875,-200){\makebox(0,0)[lb]{{\SetFigFont{12}{14.4}{\rmdefault}{\mddefault}{\updefault}$\partial\Omega_c$}}}
\put(825,6600){\makebox(0,0)[lb]{{\SetFigFont{12}{14.4}{\rmdefault}{\mddefault}{\updefault}$\partial\Omega_c$}}}
\put(2250,3525){\makebox(0,0)[lb]{{\SetFigFont{12}{14.4}{\rmdefault}{\mddefault}{\updefault}$\partial\Omega_i$}}}
\put(-150,3675){\makebox(0,0)[lb]{{\SetFigFont{12}{14.4}{\rmdefault}{\mddefault}{\updefault}$\partial\Omega_i$}}}
\put(1950,1050){\makebox(0,0)[lb]{{\SetFigFont{12}{14.4}{\rmdefault}{\mddefault}{\updefault}$J_{in}$}}}
\put(1050,5425){\makebox(0,0)[lb]{{\SetFigFont{12}{14.4}{\rmdefault}{\mddefault}{\updefault}$J_{out}$}}}
\end{picture}
}
  \caption{Typical superconducting sample. The arrows denote the direction of the current flow ($J_{in}$  for the inlet, and $J_{out}$ for the outlet).}
  \label{fig:1}
\end{figure}
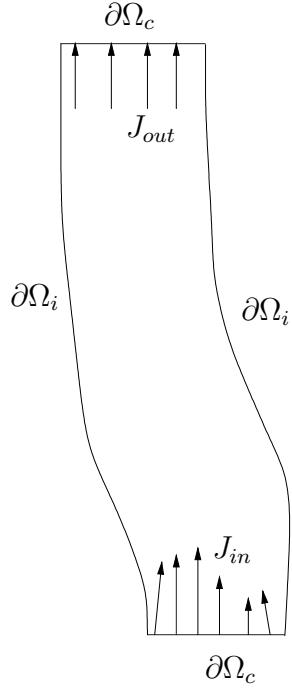
\nopagebreak
Figure 1 presents a typical sample with properties (R1) and
(R2), where the current flows into the sample from one connected
component of $\partial\Omega_c$, and exits from another part, disconnected from
the first one. Most wires would fall into the above class of domains.

The system \eqref{eq:1} is invariant to the gauge transformation
\begin{equation}
\label{eq:4}
   A'= A + \nabla\omega\,, \; \phi'= \phi-
  \frac{\partial\omega}{\partial t}\,, \; \psi' = \psi e^{i\omega} \,.
\end{equation}
We thus choose, as in \cite{alhe14}, the Coulomb gauge, i.e., we assume
\begin{equation}
  \label{eq:5}
  \begin{cases}
    \Div A =0  & \text{ in } (0,+\infty)\times \Omega\,, \\
    A\cdot\nu =0 & \text{ on }(0,+\infty)\times  \partial\Omega\,,
  \end{cases}
\end{equation}
where the divergence is computed with respect to the spatial coordinates only.

From (\ref{eq:25}b), (\ref{eq:25}d), and (\ref{eq:25}f), we know that $\curl
A$ is constant on each connected component of $\partial \Omega_i$. Let then
$\{\partial\Omega_{i,j}\}_{j=1}^{2}$ denote the set of connected components of
$\partial\Omega_i$. We can write, for $j=1,2\,$,
\begin{equation}
\label{eq:6}
\curl A|_{\partial\Omega_{i,j}} =
h_j \,  \kappa\,,
\end{equation}
where  $h_1$ and $h_2$ are constants.

Note that $h_1$ and $h_2$ can be determined from
$J$ and $h_{ex}$ via the formula \cite{alhe14}
\begin{equation}
\label{eq:7}
   h_j= h_{ex}  -  \dashint_{\pa \Omega} \, |\Gamma(\tilde x, x_j)| \,
   J(\tilde x)  ds (\tilde x)\, \mbox{ for any } x_j  \in \pa
   \Omega_{i,j}\,, \quad j=1,2\,,
\end{equation}
where $\Gamma(x,x_0)$ is the portion of $\partial\Omega$ connecting $x_0$ and $x$ in
the positive trigonometric direction.  We assume that $h_{ex}$ and $J$  are such
that
\begin{equation}
  \label{eq:8}
h_1h_2 <0 \,,
\end{equation}
and without any loss of generality we can assume $h_2>0$.  Let
\begin{equation}
  \label{eq:9}
h=\max(|h_1|,|h_2|)\,.
\end{equation}
We assume that
\begin{equation}
\label{hyph}
h>1
\end{equation}
and distinguish below between two
  different cases
\begin{subequations}
\label{eq:10}
  \begin{equation}
 \label{eq:5a}
  1<h\leq \frac{1}{\Theta_0} \quad \mbox{or}  \quad \frac{1}{\Theta_0}<h\,, \tag{\ref{eq:10}a,b}
\end{equation}
\end{subequations}
where $\Theta_0$ is given by \eqref{eq:20} ($\Theta_0\sim 0.59$).

In \cite{alhe14} we have established the  global existence of solutions for
\eqref{eq:1}, such that
\begin{align*} \label{prop1coul}
 &\psi_c\in C([0,+\infty);W^{1+\alpha,2}(\Omega,\C))\cap
 H^1_{\loc}([0,+\infty);L^2(\Omega,\C))\,, \, \forall \alpha <1\,, \\
 &A_c\in C([0,+\infty); W^{1,p}(\Omega,\R^2))\cap
 H^1_{\loc}([0,+\infty);L^2(\Omega,\R^2))\,, \forall p\geq 1\,, \\
&\phi_c\in L^2_{\loc}([0,+\infty); H^1(\Omega))\,.
\end{align*}

We next define, as in \cite[Subsection 2.2]{alhe14}, (in a slightly
different manner as the definition here is $c$-independent)  the
normal fields. They are defined as the weak solution -- $(\phi_n,A_n)\in
H^1(\Omega)\times H^1(\Omega,\R^2)$ -- of \eqref{eq:5} and
\begin{subequations}
\label{eq:11}
  \begin{empheq}[left={\empheqlbrace}]{alignat=2}
    &\, \curl^2A_n +  \nabla\phi_n = 0 \qquad &\text{in } &\Omega\,, \\
    & - \frac{\partial\phi_n}{\partial\nu} = J   \qquad &\text{on } &\partial\Omega\,, \\
    & \dashint_{\partial\Omega}\curl A_n\,ds = h_{ex}\,, && \\
    &\int_\Omega \phi_n \,dx =0\, . &&
\end{empheq}
\end{subequations}
Note that $(0,\kappa A_n,c\kappa\phi_n)$ is a steady-state solution of \eqref{eq:1}.
By \eqref{eq:11} we have (cf. \cite{alhe14})
\begin{equation}
\label{eq:12}
  \begin{cases}
    \Delta B_n = 0 & \text{in } \Omega\,, \\
    \frac{\partial B_n}{\partial\tau} = J & \text{on } \partial\Omega\,, \\
   \dashint_{\partial\Omega}B_n(x)\,ds = h_{ex} \,. &
  \end{cases}
\end{equation}
By taking the
divergence of (\ref{eq:11}a) we also  obtain
\begin{subequations}
\label{eq:11b}
 \begin{empheq}[left={\empheqlbrace}]{alignat=2}
    & \; \Delta\phi_n=0 & \qquad \text{in } \Omega\,, \\
     &- \frac{\partial\phi_n}{\partial\nu} = J   & \qquad\text{on } \partial\Omega\,, \\
   &\;\int_\Omega  \phi_n \,dx =0\, . &
  \end{empheq}
\end{subequations}
We recall from \cite[(2.16) and (2.17)] {alhe14} that
\begin{equation}\label{reguphi}
 B_n\in W^{2,p}(\Omega)\,,\, \phi_n \in W^{2,p}(\Omega) \,,\, \forall p>1\,.
 \end{equation}

 We focus attention in this work on the exponential decay of $\psi$ in
 regions where $|B_n|>1$. For steady-state solutions of \eqref{eq:1}
 in the absence of electric current ($J=0$) we may set $\phi\equiv0$ and the
 magnetic field is then constant on the boundary. The exponential
 decay of $\psi$ away from the boundary has been termed ``surface
 superconductivity'' and has extensively been studied (cf.
 \cite{fohe09} and the references within). More recently, the case of
 a non-constant magnetic field has been studied as well
 \cite{attar2014ground,heka15}. In these works $\phi$ still identically
 vanishes but nevertheless $\nabla B_n\neq0$ in view of the presence of a
 current source term $\curl h_{ex}$ in (\ref{eq:1}b).  In particular
 in \cite{heka15} it has been established, in the large $\kappa$ limit for
 the case $1\ll h\ll\kappa$ that $\psi$ is exponentially small away from
 $B_n^{-1}(0)$.

 In the absence of electric current the time-dependent case is of
 lesser interest, since every solution of \eqref{eq:1} converges to a
 steady-state solution  \cite{lid97, feta01}. This result has been obtained in \cite{feta01} by
 using the fact that the Ginzburg-Landau energy functional is a
 Lyapunov function in this case. In contrast, when $J\neq0$ this
 property of the energy functional is lost, and convergence to a
 steady-state is no-longer guaranteed. In \cite{alhe14} the global
 stability of the normal state has been established for $h=\OO(\kappa)$
 in the large $\kappa$ limit. In the present contribution, for the same
 limit, we explore the behavior of the solution for $1<h\ll\kappa$ , and
 establish exponential decay of $\psi$ in every subdomain of $\Omega$ where
 $|B_n|>1$. We do that for both steady-state solutions (whose
 existence we need to assume) and time-dependent ones. We also study
 the large-domain limit, where we obtain weaker results for
 steady-state solutions only.

Let,  for $j=1,2$,
    \begin{equation}
\label{eq:13}
    \omega_{j}  = \{x\in \Omega \, : \, (-1)^jB_n(x)>1 \} \,.
  \end{equation}
Our first theorem concerns  steady state solutions and their
exponential decay, in certain subdomains of $\Omega$, in the large $\kappa$ limit.
\begin{theorem}
\label{mainsteadystate}
 Let for $\kappa \geq 1$, $(\psi_\kappa,A_\kappa,  \phi_\kappa)$ be a time-independent solution of \eqref{eq:1}.
  Suppose that for some $j\in\{1,2\}$ we have
  \begin{equation}
\label{eq:14}
    1<|h_j| \,.
  \end{equation}
  Then,  for any compact  set $K\subset \omega_j \cup \partial \Omega_c $, there
  exist $C>0$,  $\alpha >0$, and
  $\kappa_0\geq 1 $, such that for any  $\k\geq \k_0$ we have
\begin{equation}
\label{eq:15}
\int_{K} |\psi_\kappa(x)|^2 \,dx \leq C  e^{- \alpha \kappa} \,.
\end{equation}
If, in addition,
\begin{equation}
    \frac{1}{\Theta_0}<|h_j| \,,
  \end{equation}
  then (\ref{eq:15}) is satisfied for any compact subset $K\subset\overline{\omega_j}$.
\end{theorem}
In addition to the above exponential decay, which is limited to the
region where the normal magnetic field is large, we establish a weaker
decay of $\psi_\kappa$ in the entire domain.
\begin{proposition}
\label{entire}
  Under the assumptions \eqref{eq:2}-\eqref{eq:10} there exists
  $C(J,\Omega)>0$ such that, for $\kappa \geq 1$,
  \begin{equation}
\label{eq:16}
\| \psi_\kappa \|_2 \leq C(J,\Omega) \, (1+c^{-1/2})^{1/3}\kappa^{-1/6}  \,.
  \end{equation}
\end{proposition}
Theorem \ref{mainsteadystate} is extended to the time dependent case
in the following way.
\begin{theorem}
  \label{main}
  Let $(\psi_\kappa,A_\kappa, \phi_\kappa)$ denote a time-dependent solution of
  \eqref{eq:1}. Assuming $c=1$, under the conditions of Theorem
  \ref{mainsteadystate} and \eqref{condinit}, for any compact set $K\subset
  \omega_j\cup \partial \Omega_c$ there exist $C>0$ , $\alpha >0$, and $\kappa_0\geq 1$, such
  that for any $\k\geq \k_0$ we have
\begin{equation}
\label{eq:17}
\limsup_{t\to\infty}\int_{K}  |\psi_\kappa(t,x)|^2 \,dx \leq C   e^{- \alpha \kappa}  \,.
\end{equation}
\end{theorem}

Finally, we consider steady-state solutions of (\ref{eq:1}) in the
large domain limit, i.e., we set $\kappa=c=1$ and stretch $\Omega$ by a factor
of $R \gg 1 $.  Let $\Omega^R$ be the image of $\Omega$ under the map $x\to Rx$.
We consider again steady-state solutions of \eqref{eq:1}.
\begin{subequations}
\label{eq:18}
\begin{alignat}{2}
& \Delta_A\psi + \psi   \left( 1 - |\psi|^{2} \right)-i\phi\psi =0 & \quad \text{ in } \Omega^R\, ,\\
 & \curl^2A +  \nabla\phi  =  \Im(\bar\psi\, \nabla_A\psi) & \quad \text{ in }  \Omega^R\,,\\
  &\psi=0 &\quad \text{ on }  \partial\Omega_c^R\,, \\
 &\nabla_A\psi\cdot\nu=0 & \quad \text{ on }  \partial\Omega_i^R\,,\\
 & \frac{\partial\phi}{\partial\nu}= \frac{F(R)}{R}J &\quad \text{ on } \partial\Omega_c^R\,,\\
& \frac{\partial\phi}{\partial\nu}= 0 &\quad \text{ on } \partial\Omega_i^R\,,\\[1.2ex]
&\dashint_{\partial\Omega^R}\curl A(x)\,ds = F(R)h_{ex}\,.
\end{alignat}
\end{subequations}
In the above $F(R)=R^{\gamma}$ for some $0<\gamma<1$. We study the above
problem in the limit $R\to\infty$. Assuming again (\ref{eq:2})-(\ref{eq:10})
we establish the following result.
\begin{proposition}
\label{largedomain}
 Let $(\psi,A,\phi)$ denote a solution of \eqref{eq:18}.
 Then,  there exists a compact  set $K\subset\Omega$, $C>0$, $R_0>0$, and $\alpha >0$,
such that for any  $R>R_0\,$ we have
\begin{equation}
\label{eq:126a}
\int_{K_R} |\psi (x)|^2 \,dx \leq C  e^{-\alpha R} \,,
\end{equation}
where $K_R$ is the image of $K$ under the map $x\to Rx$.
\end{proposition}
Note that $h_{ex}$ must be of $\OO(J)$, otherwise $(0,A_n,\phi_n)$ would
be the unique solution.

Physically \eqref{eq:126a} demonstrates that there is a
significant portion of the superconducting sample which remains,
practically, at the normal state, for current densities which may be
very small. This result stands in
contrast with what one finds in standard physics handbooks
\cite{poetal99} where the critical current density, for which the
fully superconducting state looses its stability, is tabulated a
material property. However, our results suggest that the critical
current depends also on the geometry of the superconducting sample. In
fact, according to Proposition \ref{largedomain}, this current density
must decay in the large domain limit. In two-dimensions, our result
suggests that one should search for a critical current (and not
current density), whereas in three-dimensions a density with respect
to cross-section circumference (instead of area) should be obtained.

We note that Proposition \ref{largedomain} is certainly not
optimal. In fact, we expect the following conjecture to be true.
\begin{conjecture}
  \label{conj}
Under the  conditions of Proposition \ref{largedomain}, for any
compact  set $K\subset\Omega\setminus B_n^{-1}(0)$, there exist $R_0>0$, $C>0$, and $\alpha >0$,
such that for any  $R>R_0$ \eqref{eq:126a} is satisfied.
\end{conjecture}

The rest of this contribution is organized as follows.  In the next
section we establish some preliminary results related to the
eigenvalues of the magnetic Laplacian in the presence of
Dirichlet-Neumann corners. We use these results in Section 3 where we
establish Theorem \ref{mainsteadystate} and Proposition~\ref{entire}.
In Section 4 we consider the time-dependent problem and establish, in
particular, Theorem~\ref{main}. Finally, in the last section, we
obtain some weaker results for steady-state solutions of (\ref{eq:1})
in the large domain limit.

\section{Magnetic Laplacian Ground States}
\label{sec:2}
In this section, we analyze the spectral properties of the Schr\"odinger
operator with constant magnetic field in a sector. The Neumann problem
has been addressed by V. Bonnaillie-No\"el in \cite{Bon}. In the sequel
we shall need, however, a lower bound for the ground state energy of
the above operator on a Dirichlet-Neumann sector, i.e., a Dirichlet
condition is prescribed on one side of the sector and the magnetic
Neumann condition on the other side. We begin by the following
auxiliary lemma whose main idea has been introduced to us by M.~Dauge
\cite{da14}.  Hereafter the norms in the Lebesgue spaces $L^p(\O)$,
$L^p(\O,\mathbb  R^2)$ and $L^p(\O,\mathbb  C)$ will be denoted by
$\|\cdot\|_{L^p(\O)}$ or $\|\cdot\|_p$, and the norms in the Sobolev spaces
$W^{k,p}(\O,\mathbb R)$, $W^{k,p}(\O,\mathbb  R^2)$ and $W^{k,p}(\O,\mathbb  C)$
will be denoted by $\|\cdot\|_{W^{k,p}(\O)}$ or $\|\cdot\|_{k,p}$.
 \begin{lemma}
\label{lem:Dirichlet-Neumann}
 Let $S_\alpha$ denote an infinite sector of angle $\alpha \in (0,\pi]$, i.e.,
 \begin{displaymath}
   S_\alpha = \{ (x,y)\in\R^2 \,: \, 0<\arg(x+iy)<\alpha \}\,.
 \end{displaymath}
Let further
 \begin{displaymath}
   \Hg_\alpha = \{ u\in H^1(S_\alpha) \,: \, u(r\cos\alpha,r\sin\alpha)=0\,, \; \forall r>0\}\,,
 \end{displaymath}
and
\begin{displaymath}
  \Theta^{DN}_\alpha  = \inf_{\begin{subarray}{c}
      u\in \Hg_\alpha \\
      \|u\|_2 =1
    \end{subarray}}
\int_{S_\alpha} |(\nabla-iF) u|^2 \,dx \,,
\end{displaymath}
where $F$ is a magnetic potential satisfying $\curl F = 1$ in
$S_\alpha$. Then,
 \begin{equation}
   \label{eq:19}
  \Theta^{DN}_\alpha =\Theta_0\,,
 \end{equation}
where
  \begin{equation}
\label{eq:20}
  \Theta_0  = \inf_{\begin{subarray}{c}
 u\in H^1(S_\pi)  \\
      \|u\|_2 =1
    \end{subarray}}
\int_{S_\pi} |(\nabla-iF)u|^2 \,dx \,.
\end{equation}
\end{lemma}
\begin{proof}
  Let $0<\alpha_1<\alpha_2\leq\pi$ and $u\in\Hg_{\alpha_1}$. Let further
  $\tilde{u}\equiv u$ in $S_{\alpha_1}$ and $\tilde{u}\equiv0$ in
  $S_{\alpha_2}\setminus S_{\alpha_1}$. Clearly $\tilde{u}\in\Hg_{\alpha_2}$, and hence it
  follows that $\Theta^{DN}_{\alpha_1}\geq\Theta^{DN}_{\alpha_2}$. Consequently,
  \begin{displaymath}
      \Theta^{DN}_\alpha \geq \Theta^{DN}_\pi\,, \quad \forall\alpha\leq\pi.
  \end{displaymath}
  From the definition of $\Theta_0$ it follows, however, that
  $\Theta^{DN}_\pi\geq\Theta_0$, and hence, $ \Theta^{DN}_\alpha \geq \Theta_0$ for all
  $\alpha\in(0,\pi]$. The proof of \eqref{eq:19} now easily follows 
    from the proof of Persson's Theorem \cite[Appendix B]{fohe09}
    providing the upper bound $\Theta_0$ for the essential spectrum of the
    magnetic Dirichlet-Neumann Laplacian in $S_\alpha$.
\end{proof}

Let $\Dg\subset\Omega$ have a smooth boundary, except at the corners of $\partial\Omega$. As
in \cite{alhe14} we let
\begin{equation}
  \label{eq:21}
    \mu_\epsilon(\A,\Dg) = \inf_{
    \begin{subarray}{c}
      u\in \Hg(\Dg) \\
      \|u\|_2 =1
    \end{subarray}} \int_\Dg |\epsilon\nabla-i\A u|^2 \,dx \,,
\end{equation}
wherein
\begin{displaymath}
   \Hg(\Dg) = \{ u\in H^1(\Dg)\,:\, u=0 \text{ on } \partial\Dg\setminus (\partial\Dg\cap \partial\Omega_i)\,\} \,.
\end{displaymath}
Let
\begin{displaymath}
  \Sg = \overline{\partial\Omega_c}\cap\overline{\partial\Omega_i}\cap\overline{\partial\Dg}
\end{displaymath}
denote the corners of $\Omega$ belonging to  $\partial\Dg$.  Following \cite{boda06}
we set for a given magnetic potential $\mathcal A \in C^1(\overline{\Omega},\mathbb  R^2)$,
\begin{displaymath}
  b=\inf_{x\in\Dg}|\curl \A| \quad ; \quad b^\prime = \inf_{x\in\partial\Dg\cap\partial\Omega_i}|\curl \A|\,.
\end{displaymath}
The following proposition is similar to a result in \cite{boda06}
obtained for a Neumann boundary condition. Here we treat a
Dirichlet-Neumann boundary condition and allow, in addition, some
dependence of the magnetic potential on the semi-classical parameter.
\begin{proposition}
  Let $a\in W^{1,\infty}(\Dg,\mathbb  R^2)$. There exist $C>0$ and $\epsilon_0>0$ such that for all
$0<\epsilon<\epsilon_0$ we have
\begin{equation}
  \label{eq:22}
 \mu_\epsilon(\A+\epsilon^{1/2}a,\Dg) \geq \epsilon \min(b,\Theta_0b^\prime) [1-C(1+\|\nabla a\|_\infty^2)\epsilon^{1/3}]\,.
\end{equation}
\end{proposition}
\begin{proof}
  The case $b=0$ being trivial, we assume that $b>0$.  We begin by
  introducing for any $\epsilon >0$ a partition of unity (cf. also
  \cite{hemo96}), i.e., families $\{\eta_i\}_{i=1}^K\subset C^\infty(\Omega)$, and
  $\{x_i\}_{i=1}^K\subset\Dg$ satisfying
\begin{displaymath}
  \sum_{i=1}^K\eta_i^2=1 \quad ; \quad {\rm supp}\,\eta_i \subset B(x_i,\epsilon^{1/3})  \quad
  ; \quad \sum_{i=1}^K|\nabla\eta_i|^2 \leq \frac{C}{\epsilon^{2/3}} \,,
\end{displaymath}
where $C>0$ is independent of $\epsilon$.

It can be easily verified that for any $u\in H^1(\Omega,\C)$
\begin{equation}
\label{eq:23}
  \|(\epsilon\nabla-i[\A+\epsilon^{1/2}a])u\|_2^2=\sum_{i=1}^K\big[
  \|(\epsilon\nabla-i[\A+\epsilon^{1/2}a])(\eta_iu)\|_2^2 - \epsilon^2\|u\nabla\eta_i\|_2^2 \big]\,.
\end{equation}
We now set
\begin{displaymath}
  v_i=\eta_iu\exp(-i\epsilon^{-1/2}a(x_i)\cdot x) \,,
\end{displaymath}
to obtain that
$$\aligned
& \|(\epsilon\nabla-i[\A+\epsilon^{1/2}a])(\eta_iu)\|_2^2\\
=&\|(\epsilon\nabla-i[\A+\epsilon^{1/2}(a-a(x_i))])v_i\|_2^2\\
 \geq&
(1-\epsilon^{1/3})\|(\epsilon\nabla-i\A)v_i\|_2^2 - \epsilon^{2/3}\|a-a(x_i))v_i\|_2^2\\
 \geq&
(1-\epsilon^{1/3})\mu_\epsilon(\A,\Dg) \|v_i\|_2^2  - C\epsilon^{4/3}\|v_i\|_2^2\,.
\endaligned
$$
Substituting the above into \eqref{eq:23} yields
\begin{equation}
\label{eq:24}
   \|(\epsilon\nabla-i[\A+\epsilon^{1/2}a])u\|_2^2 \geq (1-\epsilon^{1/3})\mu_\epsilon(\A,\Dg) \|u\|_2^2  - C\epsilon^{4/3}(1+\|\nabla a\|_\infty^2)\|u\|_2^2\,.
\end{equation}
By following the same steps of the proof  of Theorem 7.1 in \cite{boda06}
we can establish that
\begin{equation}\label{eq:74a}
 \mu_\epsilon(\A,\Dg) \geq   \epsilon \min(b,\Theta_0b^\prime,\Theta^{DN}_\alpha\inf_{x\in\Sg}|\curl \A| )
 (1-C\epsilon^{1/2})\,.
\end{equation}
The lemma now follows from \eqref{eq:19}, \eqref{eq:24}, and
\eqref{eq:74a}.
\end{proof}

\section{Steady-State Solutions}\label{sec:3}

We begin by considering steady-state solutions of \eqref{eq:1}
$(\psi_k,A_\kappa,\phi_\kappa)\in H^1(\Omega,\C\times\R^2\times\R)$ in
the limit $\kappa\to + \infty$. Hence, we look at the
system \cite[Section 5]{alhe14}
\begin{subequations}
\label{eq:25}
\begin{alignat}{2}
  - & \nabla_{\kappa A_\kappa}^2\psi_\kappa+i\kappa\phi_\kappa \psi_\kappa =\kappa^2(1-|\psi_\kappa|^2)\psi_\kappa & \quad \text{ in } \Omega \,,\\
  - &\curl^2A_\kappa+ \frac{1}{c}\nabla\phi_\kappa  =\frac{1}{\kappa}\Im(\bar\psi_\kappa\nabla_{\kappa A_\kappa}\psi_\kappa) &
  \quad \text{ in }  \Omega\,, \\
  &\psi_\kappa =0 &\quad \text{ on }  \partial\Omega_c\,, \\
  &\nabla_{\kappa A_\kappa}\psi_\kappa\cdot\nu=0 & \quad \text{ on }  \partial\Omega_i\,, \\
  & \frac{\partial\phi_\kappa}{\partial\nu} = - c\kappa J(x) &\quad \text{ on } \partial\Omega_c\,, \\
  &\frac{\partial\phi_\kappa}{\partial\nu}=0  &\quad \text{ on }  \partial\Omega_i \,, \\[1.2ex]
  &\dashint_{\partial\Omega}\curl A_\kappa \, ds = \kappa h_{ex} \,,&
\end{alignat}
\end{subequations}
with the additional gauge restriction \eqref{eq:5}.  In the above
$(\psi_k,A_\kappa,\phi_\kappa)$ is the same as $(\psi,A,\phi)$ in \eqref{eq:1}. The
subscript $\kappa$ has been added to emphasize the limit we consider
here. We assume in addition (\ref{eq:2})-(\ref{eq:10}).  By the strong
maximum principle we easily obtain that
\begin{equation}
\label{eq:26}
  \|\psi_\kappa\|_\infty <1 \,.
\end{equation}

Let $h$ be given by
\eqref{eq:9}. It has been demonstrated in \cite{alhe14} that for
some $h_c>0$, when $h<h_c\kappa$, the normal state looses its stability.
Since we consider cases for which $1<h\ll\kappa$ it is reasonable to expect
that other steady-state solutions would exist. We note, however, that
in contrast with the case $J=0$, where the existence of steady-state
solutions can be proved using variational arguments (inapplicable in
our case), existence of steady-state solutions to \eqref{eq:25} is yet an open problem
when an electric current is applied. We shall address time-dependent
solutions in the next section.

Next we set
\begin{subequations}
\label{eq:27}
  \begin{equation}
\label{eq:28}
 A_{1,\kappa}=A_\kappa-\kappa A_n\,,\q
\phi_{1,\kappa}=\phi_\kappa-c\kappa\phi_n\,.\tag{\ref{eq:27}a,b}
\end{equation}
\end{subequations}
Set further
\begin{subequations}
\label{eq:29}
  \begin{equation}
\label{eq:30}
B_\kappa =\curl A_\kappa\,,\q
B_{1,\kappa}=\curl A_{1,\kappa} \,. \tag{\ref{eq:29}a,b}
\end{equation}
\end{subequations}
By
(\ref{eq:25}b) we then have
\begin{subequations}
\label{eq:31}
  \begin{empheq}[left={\empheqlbrace}]{alignat=2}
  &\curl B_{1,\kappa} +  \frac{1}{c}\nabla\phi_{1,\kappa}  =  \frac{1}{\kappa}\Im(\bar\psi_\kappa\, \nabla_{\kappa A_\kappa}\psi_\kappa)  &  \text{ in }
  \Omega\,,\\
   &\frac{\partial\phi_{1,\kappa}}{\partial\nu}= 0 & \text{ on } \partial\Omega\,, \\
&\dashint_{\partial\Omega}B_{1,\kappa}(x)\,ds = 0 \,. &\label{eq:42c}
\end{empheq}
\end{subequations}
Note that since $\partial B_{1,\kappa}/\partial\tau=\partial\phi_{1,\kappa}/\partial\nu=0$ on $\partial\Omega$
we must have by (\ref{eq:31}c) that
\begin{equation}
\label{eq:32}
  B_{1,\kappa}|_{\partial\Omega}\equiv 0\,.
\end{equation}
Taking the divergence of (\ref{eq:25}b) yields, with the aid of the
imaginary part of (\ref{eq:25}a),  that $\phi_\kappa $ is a weak
solution of
\begin{equation}
  \label{eq:33}
  \begin{cases}
    -\Delta\phi_\kappa + c \,|\psi_\kappa|^2 \phi_\kappa =0 & \text{ in } \Omega\,, \\
  \frac{\partial\phi_\kappa}{\partial\nu}= c\kappa J &  \text{ on } \partial\Omega_c\,,\\
 \frac{\partial\phi_\kappa}{\partial\nu}= 0 &  \text{ on } \partial\Omega_i\,.
  \end{cases}
\end{equation}
By assumption $\phi_\kappa\in H^1(\Omega)$ and hence, by \cite[Proposition
A.2]{alhe14} we obtain that $\phi_\kappa\in W^{2,p}(\Omega)$ for all $p\geq2$, hence $\phi_\kappa \in C^1(\overline{\Omega})$.  By
(\ref{eq:27}b) we then have
\begin{displaymath}
  \begin{cases}
    -\Delta\phi_{1,\kappa} + c\,  |\psi_\kappa|^2 \phi_{1,\kappa} =- \kappa c^2|\psi_k|^2 \phi_n & \text{ in } \Omega\,, \\
  \frac{\partial\phi_{1,\kappa}}{\partial\nu}= 0 &  \text{ on } \partial\Omega \,.
  \end{cases}
\end{displaymath}
Let $K=\|\phi_n\|_\infty$ and $w=\phi_{1,\kappa} +K\kappa c$.  Clearly,
\begin{displaymath}
  \begin{cases}
    -\Delta w+ c\, |\psi_\kappa|^2  w =- \kappa c^2|\psi_k|^2(\phi_n-K)\geq 0 & \text{ in } \Omega\,, \\
  \frac{\partial w}{\partial\nu}= 0 &  \text{ on } \partial\Omega \,.
  \end{cases}
\end{displaymath}
It can be easily verified that $w$ is the minimizer in $H^1(\Omega)$ of
\begin{displaymath}
  \Jg(v) = \|\nabla v\|_2^2 + c\|\psi_\kappa v\|_2^2+\kappa c^2\langle|\psi_k|^2(\phi_n-K),v\rangle \,.
\end{displaymath}
As
\begin{displaymath}
   \Jg(v_+) \leq    \Jg(v)\,,
\end{displaymath}
it easily follows that $ w\geq 0$, that is
\begin{displaymath}
  \phi_{1,\kappa}+K\kappa c \geq 0 \,.
\end{displaymath}
In a similar manner we obtain
\begin{displaymath}
  \phi_{1,\kappa}-K\kappa c \leq 0 \,,
\end{displaymath}
which together with (\ref{eq:27}b) yields
\begin{equation}
\label{eq:34}
  \|\phi_\kappa\|_\infty  \leq C(\Omega,J)\, c\, \kappa \,.
\end{equation}
We now apply again Proposition A.2 in \cite{alhe14} to obtain that for
any $p\geq2$
\begin{equation}\label{estlp}
  \|\phi_\kappa\|_{2,p} \leq C(\Omega,J)\, c\, \kappa \,.
\end{equation}
We note that all elliptic estimates must be taken with special care
since $\Omega$ possesses corners.  The necessary details (with references
therein) can be found in Appendices A and B of \cite{alhe14} .

Next we set for $\delta >0$  and $\k\geq 1$,
\begin{displaymath}
  D_\delta(\kappa) = \{ x\in\Omega \,: \,|B_\kappa(x)|>(1+\delta)\kappa \} \,,
\end{displaymath}
and
\begin{equation}
\label{eq:35}
   S_\delta = \{ x\in\Omega \,: \,|B_n(x)|>(1+\delta)\}\,.
\end{equation}
By either (\ref{eq:10}a) or (\ref{eq:10}b), it follows that for $0<\delta <
h-1$, $S_\delta\neq\emptyset\,$. Below we show that the same is true for $ D_\delta(\kappa)$.
Note that \eqref{eq:8} implies that $S_\delta$ consists of two disjoint
sets:
\begin{equation}\label{sdj}
S_{\delta} = S_{\delta,1} \cup S_{\delta,2}\,,
\end{equation}
one near $\partial\Omega_{i,1}$ denoted by $S_{\delta,1}$, and one
near $\partial\Omega_{i,2}$ denoted by $S_{\delta,2}\,$. \\
We then let
\begin{equation}
  \label{eq:36}
\CC_{\delta,j}=\partial S_{\delta,j}\setminus (\partial\Omega\cap\partial S_{\delta,j})\,,\quad j=1,2\,,
\end{equation}
and
\begin{equation}\label{Cdi}
 \CC_\delta=\CC_{\delta,1}\cup \CC_{\delta,2} \,.
\end{equation}

We can now state and prove
\begin{lemma}\label{lemma3.1}
  For any $0<\alpha<1$ there exists $ \kappa_0=\kappa_0(\O,J,\alpha)>0$ such that for
  all $\kappa \geq \kappa_0$ and $0<\delta<h-1$, we have
  \begin{equation}
 \label{eq:37}
S_{\delta+\kappa^{-\alpha}} \subset  D_\delta(\kappa) \,.
  \end{equation}
\end{lemma}
\begin{proof} {\it Step 1: Prove that
for some $C(\Omega,J)>0$
\begin{equation}
\label{eq:43}
\|A_\kappa\|_\infty \leq C\, \kappa  \,.
\end{equation}}

Taking the divergence of (\ref{eq:31}a)  yields, with the aid of  (\ref{eq:31}b),
\begin{equation}
\label{eq:38}
  \begin{cases}
    - \Delta\phi_{1,\kappa}   = - \frac{c}{\kappa}\, \Div \Im(\bar\psi_\kappa\, \nabla_{\kappa A_\kappa
    }\psi_\kappa) & \text{ in } \Omega\,, \\
     \frac{\partial\phi_{1,\kappa}}{\partial\nu}= 0 & \text{ on } \partial\Omega \,.
  \end{cases}
\end{equation}
Multiplying the above equation by $\phi_{1,\kappa }$ and integrating by parts
then yields, with the aid of \eqref{eq:26} and (\ref{eq:25}c,d),
\begin{equation}
\label{eq:39}
     \|\nabla \phi_{1,\kappa} \|_2 \leq \frac{c}{\kappa}\,  \|\nabla_{\kappa A_\kappa}\psi_\kappa\|_2 \,.
 \end{equation}
Taking the inner product of (\ref{eq:25}a) with $\psi_\kappa$ yields, after
integration by parts
\begin{equation}
\label{eq:40}
   \|\nabla_{\kappa A_\kappa}\psi_\kappa\|_2^2 = \kappa^2\, \|\psi_\kappa\|_2^2 \,.
\end{equation}
By \eqref{eq:39} we then obtain that
\begin{displaymath}
\|\nabla\phi_{1,\kappa}\|_2 \leq Cc\,.
\end{displaymath}
Since $\curl B_{1,\kappa}=\nabla_\perp B_{1,\kappa}\,$, the boundedness of $ \|\nabla
B_{1,\kappa}\|_2$ then easily follows from the above and
\eqref{eq:31}. Consequently,
  \begin{equation}
\label{eq:41}
    \frac{1}{c}\|\nabla \phi_{1,\kappa} \|_2+ \|\nabla B_{1,\kappa}\|_2 \leq C\,.
\end{equation}
Note that $\nabla\phi_{1,\kappa}$ and $\nabla_\perp B_{1,\kappa}$ are respectively the  $L^2$ projections
of $\Im(\bar\psi_\kappa\, \nabla_{\kappa A_\kappa}\psi_\kappa)$ on
\begin{displaymath}
   H^0_0(\curl, \Omega)=\{\widehat V \in L^2(\Omega,\mathbb R^2)\,:\, \curl \widehat V=0\}\,,
\end{displaymath}
and
\begin{displaymath}
  \mathcal H^0_d:= \{\widehat W \in L^2(\Omega,\mathbb R^2)\,:\, \Div \widehat
  W=0\mbox{ and } \widehat W \cdot \nu =0\mbox{ on } \pa \Omega \}\,.
\end{displaymath}

Next, we attempt to estimate $\|\nabla \phi_{1,\kappa} \|_p$ and $\|\nabla B_{1,\kappa}\|_p$ for
any \linebreak $p>2$. Since $\Omega$ is simply-connected, we may conclude
from \eqref{eq:5}, (\ref{eq:29}a) and Remark B.2 in \cite{alhe14} that
there exists for any $p>2$ a constant \linebreak $C(p,\Omega)>0$ such that
\begin{displaymath}
  \|A_\kappa\|_{1,p} \leq C\, \|B_\kappa\|_p \,,
\end{displaymath}
for all $\kappa\geq 1$.  Sobolev embeddings then imply
\begin{equation}
  \label{eq:42}
\|A_\kappa\|_\infty \leq C\,  \|B_\kappa\|_p\,.
\end{equation}
Since $ \|\nabla B_{1,\kappa}\|_2$ is uniformly bounded for all $\kappa \geq 1$, we
obtain from \eqref{eq:32}, the Poincar\'e inequality, and Sobolev embeddings
that, for any $p>2$ there exists  a constant  $C(p,\Omega)>0$ such that we have
\begin{equation}\label{estB1p}
 \|B_{1,\kappa}\|_p \leq  C(p,\Omega)\,.
\end{equation}

Hence, recalling from \eqref{reguphi} that $B_n\in L^p$ and independent
of $c$ and $\kappa$, as $J$ is independent of $\kappa$, there
exists a constant $C>0$ such that
\begin{equation}\label{estBp}
  \|B_\k\|_p=\|B_{1,\k}+\k B_n\|_p\leq C\, \k\,.
\end{equation}
Combining the above computations with \eqref{eq:42} then yields
\eqref{eq:43}.
\vspace{1ex}

{\it Step 2: Prove \eqref{eq:37}.}
\vspace{1ex}

We first rewrite (\ref{eq:25}a,c,d)  in the following form
\begin{displaymath}
  \begin{cases}
\Delta\psi_\kappa   = 2i\kappa A_\kappa\cdot\nabla_{\kappa A_\kappa}\psi_\kappa  +  |\kappa A_\kappa|^2\psi_\kappa- \kappa^2\psi_\kappa\big( 1 - |\psi_\kappa|^{2}
\big)+i\kappa\phi_\kappa\psi_\kappa  & \text{ in } \Omega\, ,\\
  \psi_\kappa =0  & \text{on }  \partial\Omega_c \,, \\
\frac{\partial\psi_\kappa}{\partial\nu}=i\kappa  (A_\kappa \cdot\nu) \psi_\kappa=0\, &  \text{ on }  \partial\Omega_i \,,\\
  \end{cases}
\end{displaymath}
where the last equality follows from (\ref{eq:5}).  By \eqref{eq:43},
\eqref{eq:34}, the fact $\|\psi_\kappa\|_\infty\leq 1$, Proposition~A.3 and
Remark~A.4 in \cite{alhe14} (note that $\psi_\kappa$ vanishes at the
corners) we obtain that for some $C(\Omega,p,J)$
\begin{displaymath}
  \|\psi_\kappa\|_{2,p} \leq C\big[\kappa^4 + \kappa^2\|\nabla_{\kappa A_\kappa}\psi_\kappa\|_p\big]\,,  \quad \forall p>2\,.
\end{displaymath}
Sobolev embedding and \eqref{eq:43} then yield
\begin{equation}
\label{eq:44}
  \|\nabla_{\kappa A_\kappa}\psi_\kappa\|_\infty \leq C\big[\kappa^4 +
  \kappa^2\|\nabla_{\kappa A_\kappa}\psi_\kappa\|_p\big] \,.
\end{equation}
We now use a standard interpolation theorem to obtain that
\begin{displaymath}
  \|\nabla_{\kappa A_\kappa}\psi_\kappa\|_p\leq \|\nabla_{\kappa A_\kappa}\psi_\kappa\|_2^{2/p} \|\nabla_{\kappa A_\kappa}\psi_\kappa\|_\infty^{1-2/p}\,.
\end{displaymath}

Substituting \eqref{eq:44} in conjunction with \eqref{eq:40} into the
above inequality then yields
\begin{equation}
\label{eq:45}
   \|\nabla_{\kappa A_\kappa}\psi_\kappa\|_p\leq C\Big[\kappa^4 + \kappa^2\|\nabla_{\kappa A_\kappa}\psi_\kappa\|_p\Big]^{1-2/p}\kappa^{2/p}\,.
\end{equation}
Suppose first that
\begin{displaymath}
 \kappa^2< \|\nabla_{\kappa A_\kappa}\psi_\kappa\|_p \,.
\end{displaymath}
Then, we have
\begin{displaymath}
  \|\nabla_{\kappa A_\kappa}\psi_\kappa\|_p \leq C\|\nabla_{\kappa A_\kappa}\psi_\kappa\|_p^{1-2/p}\kappa^{2(1-1/p)}\,.
\end{displaymath}
Hence,
\begin{equation}
\label{eq:46}
  \|\nabla_{\kappa A_\kappa}\psi_\kappa\|_p \leq C\kappa^{p-1} \,.
\end{equation}
Next, assume that
\begin{displaymath}
  \|\nabla_{\kappa A_\kappa}\psi_\kappa\|_p\leq \kappa^2 \,,
\end{displaymath}
to obtain that
\begin{displaymath}
   \|\nabla_{\kappa A_\kappa}\psi_\kappa\|_p\leq C\kappa^{4-6/p} \,.
\end{displaymath}
From the above, together with \eqref{eq:46} we easily conclude that,
for any $2<p\leq3$, there exists a constant $C$ such that
\begin{equation}
\label{eq:47}
    \frac{1}{\kappa}\|\nabla_{\kappa A_\kappa}\psi_\kappa\|_p\leq C\kappa^{3(1-2/p)}\,.
  \end{equation}

  To continue, we need a $W^{1,p}$ estimates for the solution of $\Delta
  u=f$ where $f\in W^{-1,p}$.  We thus apply \cite[Theorem 7.1]{gima12},
which is valid for any domain which is bilipschitz equivalent to the
unit cube, to \eqref{eq:38}.  This yields that, for some $C(p,\Omega)>0$,
we have
\begin{equation}\label{reglp}
  \|\nabla\phi_{1,\kappa}\|_p \leq  C \frac{c}{\kappa} \|\nabla_{\kappa A_\kappa}\psi_\kappa\|_p\,,  \q  \forall\kappa \geq 1\,,\; 2<p\leq3\,.
\end{equation}
From \eqref{eq:47} and \eqref{eq:31} we then obtain that
 \begin{equation}
\label{eq:48}
   \frac{1}{c}\|\nabla \phi_{1,\kappa} \|_p + \|\nabla B_{1,\kappa} \|_p \leq C\,  \kappa^{3(p-2)/p}\,.
\end{equation}
Upon \eqref{eq:48} and \eqref{eq:32} we use the Poincar\'e inequality
together with Sobolev embeddings to conclude that
\begin{equation}
\label{B1k}
  \|B_{1,\kappa}\|_\infty  \leq C \kappa^{3(p-2)/p}\,,\quad  \forall\kappa>1\,,\; 2<p\leq3\,.
\end{equation}

Let $x\in S_{\delta+\kappa^{-\alpha}}\,$, namely
\begin{displaymath}
|B_n(x)|>(1+\delta+\kappa^{-\alpha})\,.
\end{displaymath}
From \eqref{B1k}, for some $C(p,\Omega,J)>0$ we have that
\begin{displaymath}
  \aligned
  |B_\kappa(x)|>&\kappa|B_n(x)| - |B_{1,\kappa}(x)| \geq(1+\delta+\kappa^{-\alpha})\kappa  - C \kappa^{3(p-2)/p}\\
  =&(1+\delta)\k+[\kappa^{1-\alpha}- C \kappa^{3(p-2)/p}].
  \endaligned
\end{displaymath}
By choosing
\begin{displaymath}
  2<p<\min\Big(3,\frac{6}{2+\alpha}\Big)\,,
\end{displaymath}
we have $\kappa^{1-\alpha}- C \kappa^{3(p-2)/p}>0$ for sufficiently large $\kappa\,$. Thus
$$|B_\kappa(x)|>(1+\delta)\kappa\,,
$$
and hence $x\in D_\delta(\kappa)$. Consequently, $S_{\delta+\kappa^{-\alpha}}\subset  D_\delta(\kappa)\,$.
\end{proof}

As a byproduct of the proof, we also obtain
\begin{proposition}
  For any $2<p\leq 3 $, there exists $\kappa_0\geq 1$ and $C>0$ such that
  \begin{equation}
    \label{eq:49}
\|A_{1,\kappa}\|_{2,p} \leq C \kappa^{3(p-2)/p}\,,\quad \forall \kappa \geq \kappa_0\,.
  \end{equation}
\end{proposition}
\begin{proof}
  The proof follows immediately from \eqref{eq:31}, \eqref{eq:48}, and
  Proposition B.3 in \cite{alhe14}.
\end{proof}

We can now prove the following semi-classical Agmon estimate for
$\psi_\kappa$, establishing  that it must be exponentially small in $S_\delta$.
\begin{proposition}
Suppose that $h$ satisfies (\ref{eq:10}b). Let then $j\in\{1,2\}$ be such that
$h_j>1/\Theta_0$.  There exist $C>0$  and
$\delta_0>0$, such that, for any $0<\delta\leq \delta_0\,$, some $\kappa_0(\delta)$ can be found,
for which
  \begin{equation}
 \label{eq:50}
\kappa \geq \kappa_0(\delta) \Rightarrow\int_{S_{\delta,j}} \exp\Big(\delta^{1/2}\kappa \,
  d(x,\CC_{\delta,j})\Big) |\psi_\kappa|^2 \,dx \leq \frac{C}{\delta^{3/2}} \,,
  \end{equation}
where $S_{\delta,j}$ is introduced in \eqref{sdj} and  $\mathcal C_{\delta,j}$ in  \eqref{eq:36}.
\end{proposition}
\begin{proof}
For $\delta >0$,  let  $\eta\in C^\infty(\Omega,[0,1])$ satisfy
\begin{equation}
\label{eq:51}
    \eta(x)=
    \begin{cases}
      1 & x\in S_{\delta,j}\,, \\
      0 & x\in\Omega\setminus S_{\delta/2,j} \,.
    \end{cases}
  \end{equation}
  By \eqref{eq:12} and \eqref{reguphi}, it follows that $\nabla B_n$ is
  bounded and independent of both $\delta$ and $\kappa$.  Consequently,
there exists a constant $C_1 >0$ such that
  \begin{displaymath}
    d(\CC_{\delta,j},\CC_{\delta/2,j}) \geq \frac{\delta}{C_1}\,,
  \end{displaymath}
and hence,  for some $C(\Omega,J)$ and all $0<\delta<\delta_0$ we can choose  $\eta$ such  that
  \begin{displaymath}
   |\nabla\eta| \leq \frac{C}{\delta} \,.
  \end{displaymath}
Let further $$\zeta=\chi\, \eta$$ where
\begin{displaymath}
  \chi=
  \begin{cases}
    \exp(\alpha_\delta\kappa d(x,\CC_{\delta,j})) &\text{if } x\in S_{\delta,j}\,, \\
    1 &\text{if } x\in\Omega\setminus S_{\delta,j}\,.
  \end{cases}
\end{displaymath}
We leave the determination of $\alpha_\delta$ to a later stage.  We further
define, for any $r\in (0,r_0)$,  $\eta_r\in C^\infty(\Omega,[0,1])$ and $\tilde \eta_r\in C^\infty(\Omega,[0,1])$ such that
\begin{equation}
\label{eq:52}
     \eta_r(x)=
    \begin{cases}
      1 & d(x,\partial\Omega_i)>r \\
      0 & d(x,\partial\Omega_i)<r/2 \,,
    \end{cases}
    \qq\text{and}\qq
|\nabla\eta_r|^2 + |\nabla \tilde \eta_r|^2  \leq \frac{C}{r^2} \,,
\end{equation}
and
$$
 \eta_r^2 + \tilde{\eta}_r ^2 = 1\,.
$$

Fix $0<\alpha<1$. Multiplying (\ref{eq:25}a)  by $\zeta^2\bar{\psi}$,
integrating by parts yields for the real part
$$\aligned
& \|\nabla_{\kappa A_\kappa}(\zeta\tilde{\eta}_{\kappa^{-1/2}}\psi_\kappa)\|_2^2+
 \|\nabla_{\kappa A_\kappa}(\zeta\eta_{\kappa^{-1/2}}\psi_\kappa)\|_2^2 \\
 \leq& \kappa^2\|\zeta\psi_\kappa\|_2^2
 +\|\zeta\psi_\kappa\nabla\eta_{\kappa^{-1/2}}\|_2^2+\|\zeta\psi_\kappa\nabla\tilde{\eta}_{\kappa^{-1/2}}\|_2^2 + \|\psi_\kappa\nabla\zeta\|_2^2 \,.
\endaligned
$$
Observing that $\langle\psi_\kappa\nabla\chi,\psi_\kappa\nabla\eta\rangle=0$, we obtain
\begin{displaymath}
   \|\psi_\kappa\nabla\zeta\|_2^2 \leq  \alpha_\delta^2\kappa^2\|\psi_\kappa\zeta\|_2^2 +   \|\psi_\kappa\nabla\eta\|_2^2 \,.
\end{displaymath}
Hence,
\begin{equation}\label{eq:53}
\aligned
&  \|\nabla_{\kappa A_\kappa}(\zeta\tilde{\eta}_{\kappa^{-1/2}}\psi_\kappa)\|_2^2+
 \|\nabla_{\kappa A_\kappa}(\zeta\eta_{\kappa^{-1/2}}\psi_\kappa)\|_2^2 \\
\leq& \kappa^2\big(1+ \alpha_\delta^2+C\kappa^{-1}\big)\|\zeta\psi_\kappa\|_2^2 + \|\psi_\kappa\nabla\eta\|_2^2\,.
\endaligned
\end{equation}

We now use \eqref{eq:22} and \eqref{eq:49} to obtain, for sufficiently
small $\delta$,
\begin{equation}
\label{eq:54}
\aligned
 \|\nabla_{\kappa A_\kappa}(\zeta\tilde{\eta}_{\kappa^{-1/2}}\psi_\kappa)\|_2^2
\geq&
\kappa^4\mu_{\kappa^{-2}}(A_n+\kappa^{-1}A_{1,\kappa},S_{\delta/2})\|\zeta\tilde{\eta}_{\kappa^{-1/2}}\psi_\kappa\|_2^2\\
\geq & \kappa^2\min(\Theta_0h_j,1+\delta/2)
 [1-C\kappa^{-2/3}]\|\zeta\tilde{\eta}_{\kappa^{-1/2}}\psi_\kappa\|_2^2 \\
\geq &  (1+\frac \delta 2 )\, \kappa^2[1-C\kappa^{-2/3}] \|\zeta\tilde{\eta}_{\kappa^{-1/2}}\psi_\kappa\|_2^2  \,.
\endaligned
\end{equation}
By \cite[Theorem 2.9]{AHS} we have, since $\zeta\eta_{\kappa^{-1/2}}\psi_\kappa$ vanishes on $\partial \Omega$,
\begin{equation}
\label{eq:55}
  \|\nabla_{\kappa A_\kappa}(\zeta\eta_{\kappa^{-1/2}}\psi_\kappa)\|_2^2 \geq
  (1+\delta/2-\kappa^{-1/2})\kappa^2\|\zeta\eta_{\kappa^{-1/2}}\psi_\kappa\|_2^2\,.
\end{equation}
Consequently, by  \eqref{eq:53}, \eqref{eq:54}, and
\eqref{eq:55}, and  by choosing
$$
\alpha_\delta^2 = \frac{\delta}{4}\,$$ we obtain, that for $\kappa\geq \kappa(\delta)$, with
$\kappa(\delta)$ sufficiently large:
\begin{displaymath}
  \kappa^2 \frac{\delta}{8}  \|\zeta\psi_\kappa\|_2^2 \leq  \kappa^2 \left( \frac{\delta}{4} - \widehat C \kappa^{-\frac 12}\right)\, \|\zeta\psi_\kappa\|_2^2 \leq  \|\psi_\kappa\nabla\eta\|_2^2\,,
  \end{displaymath}
from which \eqref{eq:50} easily follows.
\end{proof}

Next we consider currents satisfying only  (\ref{eq:10}a). Let, for $j=1,2$,
 \begin{equation} \label{eq:56}
  \omega_{\delta,j}  = \{x\in \Omega \,: \, (-1)^jB_n(x)>1+\delta \;;\;d(x,\partial\Omega_i) > \delta\,\} \,,\end{equation}
and
\begin{equation} \label{Gamma}
 \Gamma_{\delta,j}=  \partial\omega_{\delta,j} \setminus\partial\Omega_c\cap\partial\omega_{\delta,j} \,.
\end{equation}
We can now state
\begin{proposition}
  Suppose that for some $j\in\{1,2\}$ we have
  \begin{equation}
\label{eq:57}
    1<|h_j| \,.
  \end{equation}
Then, there exist $C>0$, $\delta_0>0$, such that  for any $0<\delta<\delta_0$,
some $\kappa_0 (\delta)>0$ can be found, for which
 \begin{equation}   \label{eq:58}
\kappa\geq \kappa_0(\delta)\Rightarrow \int_{\omega_{\delta,j}} \exp\Big(\delta^{1/2}\kappa d(x,
\Gamma_{\delta,j})\Big) |\psi_\kappa|^2 \,dx \leq \frac{C}{\delta^{3/2}}  \,.
\end{equation}
\end{proposition}
\begin{proof}
Without loss of generality we may assume $h_j>0$; otherwise we apply
  to \eqref{eq:25} the transformation
  $(\psi_\kappa,A_\kappa,\phi_\kappa)\to(\bar{\psi_\kappa},-A_\kappa,-\phi_\kappa)$.  Let
   \begin{equation}\label{check1}
 { \check{\chi}=}
  \begin{cases}
    \exp\Big(\frac 12 \delta^{1/2}\kappa d(x,\Gamma_{\delta,j})\Big) &\text{if } x\in \omega_{\delta,j}\,,  \\
    1 &\text{if } x\in\Omega\setminus \omega_{\delta,j}\,.
  \end{cases}
\end{equation}
Let further $\eta$  and  $\eta_r$  be given by \eqref{eq:51} and
\eqref{eq:52} respectively.  Then set
\begin{equation}\label{check2}
\check{\zeta}= \eta_\delta\, \eta\, \check \chi \,.
\end{equation}
The proof proceeds in the same manner as in the previous proposition
with $\zeta$ replaced by $\check{\zeta}$ with the difference that now
$\check{\zeta} \psi_\kappa(x)$ vanishes for all $ x\in\partial \Omega_i$. Consequently,
(\ref{eq:10}a) is no longer necessary (see \eqref{eq:55}). We use (\ref{eq:10}b)
to establish that $\omega_{\delta,j}$ is not empty.
\end{proof}

We conclude this section by showing that for $\OO(\kappa)$ currents
(i.e. when $J$ is independent of $\kappa$) $\|\psi_\kappa\|_2$ must be small.
To this end we define $\Phi_n$ as the solution of (\ref{eq:11b}a,b),  and
\begin{equation}
\label{eq:59}
  \int_\Omega |\psi_\kappa|^2\Phi_n \,dx =0\, .
\end{equation}
The above condition is a natural choice as by \eqref{eq:33} we have
that
\begin{displaymath}
   \int_\Omega |\psi_\kappa|^2\phi_\kappa  \,dx =0\, .
\end{displaymath}
It can be easily verified from \eqref{eq:20} that
\begin{equation}\label{calconst}
  \Phi_n=\phi_n+ C(\kappa,c)\,,
\end{equation}
where $\phi_n$ denotes the solution of \eqref{eq:11b}.
The constant can be extracted from \eqref{eq:59}:
$$
C (\kappa,c) = - \frac{\int_\Omega \phi_n |\psi_\kappa|^2\, dx }{\int_\Omega  |\psi_\kappa|^2\, dx}\,,
$$
from which we get  the following  upper bound (independent of $\kappa$ and $c$)
\begin{equation}
\label{eq:60}
|C (\kappa,c) | \leq \|\phi_n\|_\infty < +\infty\,.
\end{equation}

\begin{proposition}
  Under Assumptions \eqref{eq:2}-\eqref{eq:10} there exists \linebreak
  $C(J,\Omega)>0$  and $\k_0>0$ such that for any $\kappa\geq \k_0$,
  \begin{equation}
    \label{eq:61}
\| \psi_\kappa \|_2 \leq C (J,\Omega) (1+c^{-1/2})^{1/3}\kappa^{-1/6}  \,.
  \end{equation}
\end{proposition}
\begin{proof}
Let $\Phi_{1,\kappa}=\phi_k-c\kappa\Phi_n$.
 An immediate consequence of \eqref{eq:33} is that
\begin{equation}
\label{eq:62}
  \begin{cases}
    -\Delta\Phi_{1,\kappa}  + c|\psi|^2_\kappa \Phi_{1,\kappa} = - c^2|\psi|^2_\kappa \kappa\Phi_n & \text{ in } \Omega\,, \\
  \frac{\partial\Phi_{1,\kappa}}{\partial\nu}= 0 &  \text{ on } \partial\Omega\,.
  \end{cases}
\end{equation}
Taking the inner product with $\Phi_n$ yields, with the aid of (\ref{eq:59}),
\begin{equation}
\label{eq:63}
   \aligned
& \kappa\| \psi_\kappa \Phi_n\|_2^2 = - \frac{1}{c^2}\langle\nabla\Phi_n,\nabla\Phi_{1,\kappa}\rangle  +\frac{1}{c}\langle \langle|\psi_\kappa|
  \Phi_n,|\psi_\kappa| (\Phi_{1,\kappa}-(\Phi_{1,\kappa})_\Omega \rangle\\
\leq& \frac{1}{c^2}\|\nabla\Phi_n\|_2\|\nabla\Phi_{1,\kappa}\|_2 + \frac{1}{c}\| \psi_\kappa
  \Phi_n\|_2 \| \psi_\kappa(\Phi_{1,\kappa}-(\Phi_{1,\kappa})_\Omega )\|_2\,,
\endaligned
\end{equation}
where $(\Phi_{1,\kappa})_\Omega $ is the average of $\Phi_{1,\kappa}$ in $\Omega$.

With the aid of (\ref{eq:41}) (note that $\nabla \phi_{1,\kappa}=\nabla \Phi_{1,\kappa}$),
the fact that \break $|\psi_\kappa|\leq1$, and the Poincar\'e inequality we then
obtain
\begin{equation}
  \label{eq:64}
 \| \psi_\kappa \Phi_n\|_2 \leq  C\kappa^{-1/2}(1+ c^{-1/2}) \,.
\end{equation}
We now set
\begin{displaymath}
  \Ug_\kappa = \{ x\in\Omega \,:\, |\Phi_n(x) |< (1+c^{-1/2})^{2/3}\kappa^{-1/3}\,\} \,.
\end{displaymath}
By \eqref{eq:59}  the level set $\Phi_n^{-1}(0)$ lies inside
$\Omega$.
Let $x_0\in\Phi_n^{-1}(\tau)$ for some $\tau\neq0$, and set
$$\Gamma_\perp=B_n^{-1}(B_n(x_0)).
$$
By (\ref{eq:11}a) $B_n$ is the conjugate harmonic function of $\Phi_n$,
and hence $\Gamma_\perp$ must be perpendicular to $\Phi_n^{-1}(\tau)$ at $x_0$. Note
that in \cite[(2.3)]{alhe14} we showed that $B_n^{-1}(\mu)$ is a simple
smooth curve connecting the two connected components of $\partial\Omega_c$ for any
$h_1<\mu<h_2$. We denote by $\tilde{\Gamma}_\perp$ the subcurve of $\Gamma_\perp$
  originating from $x_0$ in the direction where $\Phi_n$ decreases if
$\tau>0$ or increases if $\tau<0$, and terminating either on $\Phi_n^{-1}(0)$ or
on the boundary.  Clearly,
\begin{displaymath}
  |\tilde{\Gamma}_\perp|\inf_{x\in\Omega}|\nabla\Phi_n|\leq  \Big|\int_{\tilde{\Gamma}_\perp} \nabla\Phi_n  ds\Big| \leq |\tau|\,,
\end{displaymath}
where $\int_\Gamma \vec{V}$ denotes the circulation of $\vec{V}$ along the path $\Gamma$.
In \cite[\S2.3]{alhe14} we have
established that $|\nabla\Phi_n|=|\nabla B_n|>0$ in $\bar{\Omega}$. It follows that
\begin{equation}
\label{eq:65}
  d(x_0,\partial\Omega\cup\Phi_n^{-1}(0)) \leq   |\tilde{\Gamma}_\perp| \leq C|\tau|\,.
\end{equation}
Let
\begin{displaymath}
  \tilde{\Ug}_\kappa(r)= \{x\in\Omega \,| \, d(x,\partial\Omega\cup\Phi_n^{-1}(0)) \leq
  r(1+c^{-1/2})^{2/3}\kappa^{-1/3}\}\,.
\end{displaymath}
By \eqref{eq:65} we obtain that for sufficiently large $r$ there
exists $\kappa_0(r)$ such that for all $\kappa>\kappa_0$ and
$c\in\R$ we have $\Ug_\kappa \subseteq \tilde{\Ug}_\kappa(r)$. Consequently,
\begin{equation}
\label{eq:66}
  | \Ug_\kappa| \leq C(1+c^{-1/2})^{2/3}\kappa^{-1/3}(|\Phi_n^{-1}(0)|+|\partial\Omega|)\leq C(1+c^{-1/2})^{2/3}\kappa^{-1/3}\,.
\end{equation}

By \eqref{eq:64} we have that
\begin{displaymath}
  \| \psi_\kappa\|_{L^2(\Omega\setminus\Ug_\kappa)} \leq C(1+c^{-1/2})^{1/3}\kappa^{-1/6} \,,
\end{displaymath}
whereas from \eqref{eq:66} and \eqref{eq:26} we learn that
\begin{displaymath}
   \| \psi_\kappa\|_{L^2(\Ug_\kappa)} \leq C(1+c^{-1/2})^{1/3}\kappa^{-1/6}  \,.
\end{displaymath}
The proposition can now be  readily verified.
\end{proof}
An immediate conclusion is that whenever $c\kappa\gg1$, $|\psi_\kappa|$ is small. If
$c=\mathcal O(\kappa^{-1})$, $|\psi_\kappa|$ may not tend to $0$ as $\kappa\to\infty$. Further research
is necessary to establish this point.
\begin{remark}
 If, for some $0<\alpha<1$, we assume that $J=J(\cdot, \kappa)$ satisfies
$$
\|J\|\leq C\kappa^{-\alpha}\,,
$$
then (\ref{eq:63}) and
  (\ref{eq:41}) remain valid. Assuming $c=1$, and using  this time (\ref{eq:48}), we obtain
  instead of (\ref{eq:64}),   for any $0<\beta$, that
  \begin{displaymath}
     \| \psi_\kappa \Phi_n\|_2 \leq  C_\beta \kappa^{-1/2}(\kappa^{\alpha/2}+\kappa^\beta\|\psi_\kappa\|_2^{1/2}) \,.
  \end{displaymath}
Using the above, with sufficiently small $\beta$, we obtain,  similarly to the derivation of  (\ref{eq:61})
\begin{displaymath}
   \| \psi_\kappa\|_2\leq C\kappa^{-(1-\alpha)/6}\,,
\end{displaymath}
which implies
$$\|\psi_\kappa\|_2\xrightarrow[\kappa\to + \infty]{}0\,.
$$
This result stands in sharp contrast with the behavior obtained
in the absence of electric potential \cite{sase03,attar2014ground}.
\end{remark}

\section{Time-Dependent Analysis}
\label{sec:4}
In this section we return to the time-dependent problem as introduced
in \eqref{eq:1}.  For convenience we set here $$c=1\,.$$
\begin{subequations}
\label{eq:67}
\begin{alignat}{2}
 & \frac{\partial\psi_\kappa}{\partial t}-  \nabla_{\kappa A_\kappa}^2\psi_\k+i\kappa\phi_\kappa \psi_\kappa =\kappa^2(1-|\psi_\kappa|^2)\psi_\kappa & \text{ in } (0,+\infty)\times \Omega& \,,\\
  &\frac{\partial A_\kappa}{\partial t}+\nabla\phi_\kappa +\curl^2A_\kappa
  =\frac{1}{\kappa}\Im(\bar\psi_\kappa\nabla_{\kappa A_\kappa}\psi_\kappa) &
  \text{ in } (0,+\infty)\times  \Omega&\,, \\
  &\psi =0 &\text{ on } (0,+\infty)\times  \partial\Omega_c&\,, \\
&  \nabla_{\kappa A_\kappa}\psi_\kappa\cdot\nu=0 &  \text{ on } (0,+\infty)\times  \partial\Omega_i&\,, \\
&  \frac{\partial\phi_\kappa}{\partial\nu} = - \kappa J(x) & \text{ on }(0,+\infty)\times  \partial\Omega_c&\,, \\
 & \frac{\partial\phi_\kappa}{\partial\nu}=0  & \text{ on } (0,+\infty)\times  \partial\Omega_i &\,, \\[1.2ex]
&  \dashint_{\partial\Omega}\curl A_\kappa \, ds = \kappa h_{ex} &\quad \text{ on } (0,+\infty)&\,, \\
&  \psi(0,x)=\psi_0(x)  & \text{ in } \Omega&\,, \\
&A(0,x)=A_0(x) & \text{ in } \Omega & \,.
\end{alignat}
\end{subequations}
We assume again \eqref{condinit}-(\ref{eq:6}), \eqref{eq:8}, and
\eqref{eq:10}.  Since in the time dependent case $\phi_\kappa$ is determined
up to a constant in view of \eqref{eq:4}
 and \eqref{eq:5}, we can
further  impose
\begin{equation}\label{eq:68}
  \int_\Omega\phi_\kappa (t,x)\,dx =0\,, \quad \forall t>0 \,.
\end{equation}
It follows from \eqref{condinit} by the maximum
principle (see \cite[Theorem 2.6] {alhe14}) that
\begin{equation}
\label{eq:69}
  \|\psi_\kappa(t, \cdot)\|_\infty\leq 1 \,, \quad \forall t\geq0\,.
\end{equation}
We recall from \cite  [Subsection 2.4] {alhe14} the following spectral entity
\begin{subequations}
\label{eq:70}
  \begin{equation}
\lambda= \inf_{
    \begin{subarray}{2}
      V\in\Hg_d \\
      \|V\|_2 =1
    \end{subarray}} \| \curl V\|_2^2 \,,
\end{equation}
where
\begin{equation}
  \Hg_d = \big\{ V\in H^1(\Omega,\R^2)\,:\, \Div V=0 \,, V\big|_{\partial\Omega}\cdot\nu=0 \big\}  \,.
\end{equation}
\end{subequations}
We further recall from \cite [Proposition 2.5] {alhe14} that, under condition $(R_1)$ on $\partial \Omega$,
\begin{displaymath}
  \lambda=\lambda^D:=\inf_{
    \begin{subarray}{2}
      u\in H^1_0(\Omega) \\
      \|u\|_2 =1
    \end{subarray}} \| \nabla u\|_2^2>0 \,.
\end{displaymath}

We retain our definition of the normal fields $(A_n,\phi_n)$ via \eqref{eq:11}. For
the solution $(A_\k,\phi_\k)$ of \eqref{eq:67} we set
\begin{equation}\label{time-dep-AB}  \aligned
&A_{1,\kappa}(t,x)=A_\kappa(t,x)-\kappa A_n(x),\\
&\phi_{1,\kappa}(t,x)=\phi_\kappa(t,x)-\kappa\phi_n(x),\\
&B_\kappa(t,x)=\curl A_\kappa(t,x),\\
&B_{1,\kappa}(t,x)=\curl A_{1,\kappa}(t,x).
\endaligned
\end{equation}
Clearly,
\begin{subequations}
\label{eq:71}
  \begin{alignat}{2}
  \frac{\partial A_{1,\kappa}}{\partial t}+\nabla\phi_{1,\kappa}  + \curl B_{1,\kappa}
 & =  \frac{1}{\kappa}\Im(\bar\psi_\kappa\, \nabla_{\kappa A_\kappa}\psi_\kappa)  &  \text{ in }
 (0,+\infty)\times  \Omega\,,\\
   \frac{\partial\phi_{1,\kappa}}{\partial\nu}&= 0 & \text{ on } (0,+\infty)\times \partial\Omega\,, \\
\dashint_{\partial\Omega}B_{1,\kappa}(t,x)\,ds& = 0  & \text{ in } (0,+\infty) \,.
\end{alignat}
\end{subequations}

We begin by the following auxiliary estimate. We recall that \break $\|A(t,\cdot)\|_{1,2}=\|A(t,\cdot)\|_{H^1(\O,\R^2)}$.
\begin{lemma}
\label{lem:parabolic-A1}
Let $A_{1,\kappa}$ and $B_{1,\kappa}$ be defined by \eqref{time-dep-AB}.
Suppose that \linebreak $\|A_{1,\kappa}(\cdot,0)\|_2\leq M$ (where $M$ may depend on
$\kappa$). Then, under the above assumptions, there exists $ t^*(M)$ and a
constant $C=C(\Omega,t^*)>0$ such that for all $t>t^*$ and $\kappa\geq 1$ we
have
\begin{equation}
  \label{eq:72}
\|A_{1,\kappa}(t,\cdot)\|_{1,2} +  \|A_{1,\kappa}\|_{L^2(t,t+1,H^2(\Omega))} \leq C.
\end{equation}
\end{lemma}
\begin{proof}
  By \cite[Lemma 5.3]{alhe14} there exists a constant $C=C(\Omega)>0$ such that for
  sufficiently large
  $\kappa$
  \begin{displaymath}
    \|A_{1,\kappa}(t,\cdot)\|_2^2 \leq \Big(\|A_{1,\kappa}(0,\cdot)\|_2^2+\frac{C}{\kappa^2}\Big)
e^{-\lambda t} + C   \int_0^te^{-\lambda(t-\tau)}\|\psi_\kappa(\tau,\cdot)\|_2^2 \,d\tau \,.
  \end{displaymath}
From \eqref{eq:69},  we then get
   \begin{displaymath}
    \|A_{1,\kappa}(t,\cdot)\|_2^2 \leq \Big(\|A_{1,\kappa}(0,\cdot)\|_2^2+\frac{C}{\kappa^2}\Big)\,
e^{-\lambda t} + \frac{C}{\lambda}
  \end{displaymath}
We thus have
     \begin{equation}\label{anytime}
    \|A_{1,\kappa}(t,\cdot)\|_2^2 \leq  (M+C) e^{-\lambda t} + \frac{C}{\lambda}
  \end{equation}
Hence, there exists $t_0^*(M)$, such that for $t \geq t_0^*(M)$,  we have
\begin{equation}
\label{eq:73}
    \|A_{1,\kappa}(t,\cdot)\|_2 \leq \frac{2C}{\lambda} \,.
\end{equation}
Next, we apply \cite[Theorem C.1 (Formula C.4)]{alhe14} to the operator
$\mathcal L^{(1)}$ (as introduced there  in Example (4) above this theorem) to
obtain that
\begin{equation}\label{consc1}
  \aligned
& \|A_{1,\kappa}\|_{L^\infty(t_0,t_0+1,H^1(\Omega))}+
  \|A_{1,\kappa}\|_{L^2(t_0,t_0+1,H^2(\Omega))}\\& \qquad
\leq  \frac{C}{\kappa}\|\Im \{\bar{\psi}_\k\nabla_{\kappa
    A_\kappa}\psi_\kappa\}\|_{L^2(t_0,t_0+1,L^2(\Omega))} +
  C  \|A_{1,\kappa}(t_0,\cdot)\|_{1,2} \,.
\endaligned
\end{equation}
with a constant $C$ independent of $t_0$.
Since from (\ref{eq:67}a) (cf. \cite{alhe14}) we can easily get
that
\begin{equation}
\label{eq:74}
      \|\nabla_{\kappa A_\k}\psi_\k (t,\cdot)\|_2^2 \leq  \kappa^2\|\psi_\kappa(t,\cdot)\|_2^2-\frac{1}{2}
      \frac{d\|\psi_\k (t,\cdot)\|_2^2}{dt} \,,
\end{equation}
we obtain by integrating over $(t_0,t_0+1)$
\begin{equation}\label{eq:70a}
  \|\nabla_{\kappa A_\k}\psi_\k \|^2_{L^2(t_0,t_0+1,L^2(\Omega))} \leq  \kappa^2\|\psi_\kappa \|_{L^2(t_0,t_0+1,L^2(\Omega))}^2 + \frac 12 \|\psi_\kappa(t_0, \cdot)\|_2^2\,,
 \end{equation}
and note for later reference that it implies
 \begin{equation}\label{eq:70aa}
  \|\nabla_{\kappa A_\k}\psi_\k \|_{L^2(t_0,t_0+1,L^2(\Omega))} \leq   C(\Omega)\, \kappa\,.
 \end{equation}
Implementing the upper bound \eqref{eq:70a}  in \eqref{consc1}, yields
$$\aligned
& \|A_{1,\kappa}\|_{L^\infty(t_0,t_0+1,H^1(\Omega))}+ \|A_{1,\kappa}\|_{L^2(t_0,t_0+1,H^2(\Omega))} \\
\leq& C\Big[1+
  \|\psi_\kappa\|_{L^2(t_0,t_0+1,L^2(\Omega))}  +
  \frac{1}{\kappa}\|\psi_\kappa(t_0, \cdot)\|_2+\|A_{1,\kappa}(t_0,\cdot)\|_{1,2}\Big] \,.
\endaligned
$$
We next apply \cite[Theorem C.1 (Formula C.2)]{alhe14} to obtain in
precisely the same manner
$$
\begin{array}{l}
\|A_{1,\kappa}\|_{L^\infty(t_0,t_0+1,H^1(\Omega))}\\
\qquad  \leq C\Big[1+
  \|\psi_\kappa\|_{L^2(t_0-1,t_0+1,L^2(\Omega))}  +
  \frac{1}{\kappa}\|\psi_\kappa(t_0, \cdot)\|_2 + \|A_{1,\kappa}(t_0-1,\cdot)\|_2 \Big] \,.
  \end{array}
$$
The above together with \eqref{eq:69} and \eqref{eq:73} yields, for
 $t_0\geq t_0^* +1$,
\begin{equation}
\label{est-of-B}
\|A_{1,\kappa}\|_{L^\infty(t_0,t_0+1,H^1(\Omega))}+ \|A_{1,\kappa}\|_{L^2(t_0,t_0+1,H^2(\Omega))}
\leq C\,,
\end{equation}
which implies  \eqref{eq:72}, with $t^* = t_0^* (M) +1$.
\end{proof}
\begin{remark}\label{usefulrem}
  Since our interest is in the limit as $t \to +\infty$, Lemma~\ref{lem:parabolic-A1} allows us to assume in the sequel, without
  any loss of generality, that (\ref{eq:72}) is satisfied for all
  $t\geq0$. We have just to make a translation $t \mapsto t -t^*$ and to
  observe that $\psi_\kappa(t^*,\cdot)$ has the same properties as $\psi_0$.
  \end{remark}
\begin{proposition}
  \label{lem:asymp-exp}
  Let $\omega_{\delta,j}$ ($j\in\{1,2\}$) be defined in
  \eqref{eq:56}.
Suppose that
  for some $j\in\{1,2\}$ we have that
  \begin{equation}
\label{eq:75}
    1<|h_j| \,.
  \end{equation}
 Then, there exist $C>0$ and $\delta_0>0$, and, for any $0<\delta<\delta_0$,  $\kappa_0(\delta)\geq 1$ such that, for  $\kappa\geq \kappa_0(\delta)$,
\begin{equation}
\label{eq:76}
\limsup_{t\to\infty}\int_{\omega_{\delta,j} }|\psi_\kappa|^2(t,x) \,dx \leq
  \frac{C_\delta }{\kappa^2}  \,.
\end{equation}
\end{proposition}
\begin{proof} 
  Without loss of generality we may assume $h_j>0\,$; otherwise we apply
to \eqref{eq:67} the transformation
$(\psi_\kappa,A_\kappa,\phi_\kappa)\to(\bar{\psi_\kappa},-A_\kappa,-\phi_\kappa)$.
\vspace{1ex}

{\it Step 1: Let, for $n\geq1$,
\begin{displaymath}
  a_n = \|{\widehat \zeta}\psi_\kappa\|_{L^\infty(n-1, n,L^2(\Omega))} \,.
\end{displaymath}
Prove that for all $\delta \in (0,1)$  and $\kappa \geq \kappa_0(\delta)$,
\begin{equation}\label{inega}
  a_n^2 \leq  C \delta^{-3} \Big(\kappa^{-2} + \kappa^{-1} (a_{n} + a_{n-1})\Big)\,.
\end{equation}}
\vspace{2ex}

Let $\eta$ and $\eta_r$ be given by \eqref{eq:51} and \eqref{eq:52}
respectively. Then, set
\begin{equation}
\label{eq:77}
  {\widehat \zeta}=\eta \, \eta_\delta\, .
\end{equation}
Multiplying (\ref{eq:67}a) by ${\widehat \zeta}^{\,2} \bar{\psi_\kappa}$ and integrating by
 parts yields
\begin{displaymath}
  \frac{1}{2}\frac{d}{dt} \left(  \|{\widehat \zeta}\psi_\kappa (t,\cdot)\|_2^2\right)  \, +  \|\nabla_{\kappa A_\kappa}({\widehat \zeta}\psi_\kappa (t,\cdot))\|_2^2 \leq
  \kappa^2\|{\widehat \zeta}\psi_\kappa (t,\cdot)\|_2^2 + \|\psi_\kappa (t,\cdot)\, \nabla{\widehat \zeta}\|_2^2 \,.
\end{displaymath}
By \cite[Theorem 2.9] {AHS},
we have
$$\aligned
  \|\nabla_{\kappa A_\kappa}({\widehat \zeta}\psi_\kappa (t,\cdot))\|_2^2 \geq& \langle\kappa B_\kappa (t,\cdot) {\widehat \zeta}\psi_\kappa(t,\cdot),{\widehat \zeta}\psi_\kappa(t,\cdot)\rangle \\ =&
  \kappa^2\langle B_n{\widehat \zeta}\psi_\kappa(t,\cdot),{\widehat \zeta}\psi_\kappa(t,\cdot) \rangle + \langle\kappa B_{1,\kappa}(t,\cdot)\,{\widehat \zeta}\psi_\kappa(t,\cdot),{\widehat \zeta}\psi_\kappa(t,\cdot)\rangle  \\
  \geq&
  \kappa^2\Big(1+\frac{\delta}{2}\Big)\|{\widehat \zeta}\psi_\kappa(t,\cdot)\|_2^2+ \langle\kappa B_{1,\kappa}{\widehat \zeta}\psi_\kappa(t,\cdot)\,,\,{\widehat \zeta}\psi_\kappa(t,\cdot)\rangle \,.
\endaligned
$$
We can thus write
\begin{displaymath}
\begin{array}{l}
   \frac{1}{2}\frac{d}{dt} \left( \|{\widehat \zeta}\psi_\kappa(t,\cdot) \|_2^2\right) +   \frac{\kappa^2\delta}{2}\|{\widehat \zeta}\psi_\kappa (t,\cdot) \|_2^2\\ \quad  \leq
  \|\psi_\kappa (t,\cdot) \nabla\eta\|_2^2 + \|\psi_\kappa(t,\cdot) \nabla\eta_\delta\|_2^2- \langle\kappa B_{1,\kappa}(t,\cdot) {\widehat \zeta}\psi_\kappa(t,\cdot),{\widehat \zeta}\psi_\kappa(t,\cdot)\rangle \,.
  \end{array}
\end{displaymath}
Since
\begin{displaymath}
  \|\psi_\kappa (t,\cdot)\nabla\eta\|_2^2 + \|\psi_\kappa(t,\cdot) \nabla\eta_\delta\|_2^2 \leq \frac{C}{\delta^2} \,,
\end{displaymath}
we obtain that
\begin{equation}
  \label{eq:78}
  \begin{array}{l}
   \frac{1}{2}\frac{d}{dt} \left( \|{\widehat \zeta}\psi_\kappa (t,\cdot)\|_2^2 \right)
    +   \frac{\kappa^2\delta}{2}\|{\widehat \zeta}\psi_\kappa (t,\cdot)\|_2^2\\ \qquad
   \leq \frac{C}{\delta^2}  - \langle \kappa B_{1,\kappa} (t,\cdot) {\widehat \zeta}\psi_\kappa(t,\cdot)\,,\, {\widehat \zeta}\psi_\kappa (t,\cdot) \rangle \,.
\end{array} \end{equation}
From \eqref{eq:78} we can conclude that
\begin{equation}\label{eq:79}
\aligned
 \|{\widehat \zeta}\psi_\kappa(t,\cdot)\|_2^2 \leq& \|{\widehat \zeta}\psi_0\|_2^2\, e^{-\delta\kappa^2t} +
  \frac{C}{\delta^3\kappa^2}\\
  &+2 \int_0^t
  e^{-\delta\kappa^2(t-\tau)}\big|\langle\kappa B_{1,\kappa}(\tau,\cdot) {\widehat \zeta}\psi_\kappa(\tau,\cdot) ,{\widehat \zeta}\psi_\kappa(\tau,\cdot) \rangle \big|\,d\tau \,.
\endaligned
\end{equation}
 
To estimate the last term on the right-hand-side of \eqref{eq:79},  we start from
 $$ \aligned
&  \int_0^t  e^{-\delta\kappa^2(t-\tau)}\big|\langle\kappa
  B_{1,\kappa}(\tau,\cdot) {\widehat \zeta}\psi_\kappa (\tau,\cdot) ,{\widehat \zeta}\psi_\kappa(\tau, \cdot) \rangle\big|\,d\tau  \\
\leq &   \kappa\, \Big[\int_0^t  e^{-\delta\kappa^2(t-\tau)}
  {  \|B_{1,\kappa}(\tau, \cdot)\|_2^2 \,d\tau} \cdot \int_0^t  e^{-\delta\kappa^2(t-\tau)}
  \|{\widehat \zeta}{ \psi_\kappa(\tau, \cdot )\|_4^4}\,d\tau\Big]^{1/2}\,. \\
   \endaligned
  $$
With Remark \ref{usefulrem} in mind,  we use \eqref{eq:69} to obtain
$$
 \int_0^t  e^{-\delta\kappa^2(t-\tau)}
    \|B_{1,\kappa}(\tau, \cdot)\|_2^2 \,d\tau  \leq \frac{C}{\delta \kappa^2} \,.
$$
Implementing the above estimate, we obtain
$$
\begin{array}{l}
\int_0^t  e^{-\delta\kappa^2(t-\tau)}\big|\langle\kappa
  B_{1,\kappa}(\tau,\cdot) {\widehat \zeta}\psi_\kappa (\tau,\cdot) ,{\widehat \zeta}\psi_\kappa(\tau, \cdot) \rangle\big|\,d\tau\\
 \qquad
   \leq C \delta^{-\frac 12} \Big[\int_0^t  e^{-\delta\kappa^2(t-\tau)}
  \|{\widehat \zeta}{\psi_\kappa(\tau, \cdot )\|_4^4}\,d\tau\Big]^{1/2}\,.
  \end{array}
$$
To control of the right hand side we now write for $t\geq 1$
$$\aligned
& \Big[\int_0^t  e^{-\delta\kappa^2(t-\tau)}
  \|{\widehat \zeta}\psi_\kappa(\tau, \cdot )\|_4^4\,d\tau\Big]^{1/2} \\
& \leq  \Big[\int_0^{t-1}  e^{-\delta\kappa^2(t-\tau)}
  \|{\widehat \zeta}\psi_\kappa(\tau, \cdot )\|_4^4\,d\tau\Big]^{1/2} +  \Big[\int_{t-1} ^t  e^{-\delta\kappa^2(t-\tau)}
  \|{\widehat \zeta}\psi_\kappa(\tau, \cdot )\|_4^4\,d\tau\Big]^{1/2}\\
 & \leq C  \Big[\int_{t-1} ^t  e^{-\delta\kappa^2(t-\tau)}
  \|{\widehat \zeta}\psi_\kappa(\tau, \cdot )\|_2^2\,d\tau\Big]^{1/2}  + C \delta^{-\frac 12}  \kappa^{-1}  e^{- \frac{\delta }{2} \kappa^2}\\
& \leq  C \delta^{-\frac 12} \kappa^{-1} \|\widehat \zeta
  \psi_\kappa\|_{L^\infty(t-1,t,L^2(\Omega)) } + C \delta^{-\frac 12}  \kappa^{-1}  e^{-
    \frac{\delta }{2} \kappa^2} \,.
\endaligned
$$

Substituting the above into \eqref{eq:79} yields, with a new constant
$C$, for $\kappa$ large enough, and for $t\geq1\,$,
$$
\begin{array}{l}
   \|{\widehat \zeta}\psi_\kappa(t,\cdot)\|_2^2 \\ \quad   \leq C \delta^{-3}  \kappa^{-2} +  C \delta^{-1} \kappa^{-1}  e^{ - \frac{\delta}{2} \kappa^2}
    + C \delta^{-1} \kappa^{-1} \|\widehat \zeta \psi_\kappa\|_{L^\infty(t-1,t,L^2(\Omega)) }  \,.
     \end{array}
$$
From which we easily obtain \eqref{inega}.
\vspace{1ex}

{\it Step 2: Prove \eqref{eq:76}.}
\vspace{1ex}

By (\ref{eq:69}) we have 
$$
0 < a_n \leq C \,,
$$
which readily yields
$$
a_n \leq  C \delta^{-\frac 32}\kappa^{-\frac 12}\,.
$$
We improve the above estimate by reimplementing \eqref{inega}.
To this end we set
$$\widehat C := C \delta^{- \frac 32 }\,,$$ and then let
\begin{displaymath}
  a_n= \frac{\widehat C}{\kappa}\alpha_n \,.
\end{displaymath}
Substituting into (\ref{inega}) yields
\begin{equation}\label{inegaga}
  \alpha_n^2 \leq 1 + \alpha_{n-1} +\alpha_n \,.
\end{equation}
Suppose that for some $N\geq0$, we have $\alpha_N\leq 1 + \sqrt{2} \,$,  then\break $\alpha_{N+1}\leq 1 +\sqrt{2}$ and
hence $\alpha_n\leq 1 + \sqrt{2}$ for all $n\geq N$.

If $\alpha_{n-1} > 1+ \sqrt{2}$ for any $n$,  we have, with $\hat \alpha_n = \alpha_n - \frac 12\,$,
\begin{displaymath}
 \hat  \alpha_n^2\leq \frac 74 + \hat \alpha_{n-1} < \hat \alpha_{n-1}^2  \,.
\end{displaymath}
Hence, $\hat \alpha_n<\hat \alpha_{n-1}$ which means that $\hat \alpha_n$
converges as a positive decreasing sequence, and necessarily to a
limit smaller than $1/2 + \sqrt{2}$. We thus conclude that
\begin{equation}
  \label{eq:80}
\limsup \alpha_n \leq 1 + \sqrt{2} \,,
\end{equation}
and hence
\begin{equation}
\label{ineq-1} \limsup a_n \leq \frac{C (1 +\sqrt{2})}{\delta^{3/2}\kappa},
\end{equation}
from which \eqref{eq:76} can easily be  deduced.
\end{proof}

We next obtain the following improvement over (\ref{eq:70aa}) for
$\nabla_{\kappa A_\kappa}(\widehat \zeta \psi_\kappa)$.

\begin{proposition}
  Let $p \geq 2$.  For any $\delta >0$, there exists $\kappa_0 (\delta)$ and $C(\delta)$
  such that for $\kappa \geq \kappa (\delta)$ we have, with $u={\widehat \zeta}\psi_\kappa\,,
  $ the following estimate:
\begin{equation}\label{eq:89a}
  \|\nabla_{\kappa A_\k}u\|_{L^p(t_0,t_0+1; L^p(\Omega))}\leq C (\delta) \kappa^{6(1-2/p)}\,.
  \end{equation}
\end{proposition}

\begin{proof} {\it Step 1: Prove that for some $C(\delta)>0$ we have, for
  sufficiently large $\kappa$ that
\begin{equation}
  \label{eq:87}
\|u\|_{L^2(t_0-1,t_0+1,H^2(\Omega))} \leq C\kappa^3\,.
\end{equation}}
\vspace{1ex}

We rewrite (\ref{eq:67}a, c, d) in the form
\begin{displaymath}
  \begin{cases}
\frac{\partial\psi_\kappa}{\partial t}-\Delta\psi_\kappa   = -2i\kappa A_\kappa\cdot\nabla\psi_\kappa  -  |\kappa A_\kappa|^2\psi_\kappa+ \kappa^2\psi_\kappa\big( 1 - |\psi_\kappa|^{2}
\big)\\
\qq\qq\qq \qquad   - i\kappa\phi_\kappa\psi_\kappa&\text{ in } (0,+\infty)\times\O,\\
  \psi_\kappa =0   &\text{ on } (0,+\infty)\times  \partial\Omega_c \,, \\
\frac{\partial\psi_\kappa}{\partial\nu}=i\kappa A_\kappa\psi_\kappa\cdot\nu=0  &\text{ on }(0,+\infty)\times   \partial\Omega_i \,.\\
  \end{cases}
\end{displaymath}
Clearly, by our choice of $\widehat \zeta$,
\begin{displaymath}
    \begin{cases}
\frac{\partial u}{\partial t}-\Delta u
= {\widehat \zeta}\big[-2i\kappa A_\kappa\cdot\nabla\psi_\kappa -  |\kappa
A_\kappa|^2\psi_\kappa + \kappa^2\psi_\kappa\big( 1 - |\psi_\kappa|^{2}
\big)&  \\
\qquad\qq\q - i\kappa\phi_\kappa\psi_\kappa\big]  + 2\nabla{\widehat \zeta}\cdot\nabla\psi_\kappa
  + \psi_\kappa \Delta{\widehat \zeta}   & \text{ in } (0,+\infty)\times  \Omega\, ,\\
  u =0  & \text{on }  (0,+\infty)\times \partial\Omega \,.
\end{cases}
\end{displaymath}
By  \cite[Theorem C.1]{alhe14} (this time applied to the Dirichlet
Laplacian in $\Omega$)  in the interval $(t_0-1,t_0+1)$
\begin{equation}\label{eq:81}
\aligned
&\|u\|_{L^2(t_0,t_0+1,H^2(\Omega))} \\
\leq &\big\|{\widehat \zeta}\big[-2i\kappa A_\kappa\cdot\nabla\psi_\kappa -  |\kappa
A_\kappa|^2\psi_\kappa + \kappa^2\psi_\kappa\big( 1 - |\psi_\kappa|^{2}
\big)- i\kappa\phi_\kappa\psi_\kappa\big]\big\| _{L^2(t_0-1,t_0+1,L^2(\Omega,\R^2))} \\
&+ \|2\nabla{\widehat \zeta}\cdot\nabla\psi_\kappa
  + \psi_\kappa \Delta{\widehat \zeta}\|_{L^2(t_0-1,t_0+1,L^2(\Omega,\R^2))} + C\|u(t_0-1,\cdot)\|_{2} \,.
\endaligned
\end{equation}

By \eqref{eq:69} we have that
\begin{equation}
\label{eq:82}
  \|2\nabla{\widehat \zeta}\cdot\nabla\psi_\kappa  + \psi_\kappa \Delta{\widehat \zeta}\|_{L^2(t_0-1,t_0+1,L^2(\Omega))}
\leq C(1+\|\nabla\psi_\kappa\|_{L^2(t_0-1,t_0+1,L^2(\Omega))})  \,.
\end{equation}
As
\begin{displaymath}\
  \|\nabla\psi_\kappa\|_{L^2(t_0-1,t_0+1,L^2(\Omega))} \leq
  \|\nabla_{\kappa A_\kappa}\psi_\kappa\|_{L^2(t_0-1,t_0+1,L^2(\Omega))}+\|\kappa A_\kappa\psi_\kappa\|_{L^2(t_0-1,t_0+1,L^2(\Omega))}\,,
\end{displaymath}
we obtain in view of (\ref{eq:70a})  and (\ref{eq:72}) that
\begin{equation}
\label{eq:83}
  \|\nabla\psi_\kappa\|_{L^2(t_0-1,t_0+1,L^2(\Omega))} \leq C\kappa^2 \,.
\end{equation}
Substituting the above into (\ref{eq:82}) yields
\begin{equation}
  \label{eq:84}
\|2\nabla{\widehat \zeta}\cdot\nabla\psi_\kappa  + \psi_\kappa \Delta{\widehat \zeta}\|_{L^2(t_0-1,t_0+1,L^2(\Omega))} \leq C\kappa^2\,.
\end{equation}

We next observe that
\begin{displaymath}
 \|{\widehat \zeta} |\kappa A_\kappa|^2\psi_\kappa \|_{L^2(t_0-1,t_0+1,L^2(\Omega))} \leq
 \kappa^2\|A_\kappa\|_{L^4(t_0-1,t_0+1,L^\infty(\Omega,\R^2))}^2 \|{\widehat \zeta}\psi_\kappa \|_{L^2(t_0-1,t_0+1,L^2(\Omega))}\,.
\end{displaymath}
Since $\widehat \zeta$ is supported in the set $\omega_{\delta,j}$, we may use
\eqref{eq:76}, which together with Agmon's inequality \cite[Lemma
13.2]{ag65}, (\ref{time-dep-AB}), and \eqref{eq:72}, yield, for $t_0$ large enough,
\begin{equation}\label{eq:85}  \aligned
& \|{\widehat \zeta} |\kappa A_\kappa|^2\psi_\kappa \|_{L^2(t_0-1,t_0+1,L^2(\Omega))}\\
\leq& C\kappa\|A_\kappa\|_{L^\infty(t_0-1,t_0+1,L^2(\Omega,\R^2))}
\|A_\kappa\|_{L^2(t_0-1,t_0+1,H^2(\Omega,\R^2))}\\
\leq& C\kappa^3\,.
\endaligned
\end{equation}
Similarly,
\begin{equation}\label{eq:86}
\aligned
& \|{\widehat \zeta}\kappa A_\kappa\cdot\nabla\psi_\kappa\|_{L^2(t_0-1,t_0+1,L^2(\Omega))}\\
\leq& \kappa\|A_\kappa\|_{L^4(t_0,t_0+1,L^\infty(\Omega,\R^2))}
  \big[\|u\|_{L^4(t_0-1,t_0+1,H^1(\Omega))}
  +\|\psi_\kappa\nabla{\widehat \zeta}\|_{L^4(t_0-1,t_0+1,L^2(\Omega))}\big]\\
\leq&   C\kappa^2 +
  C\kappa^2\|u\|^{1/2}_{L^2(t_0-1,t_0+1,H^2(\Omega))}\|u\|^\frac 12 _{L^\infty(t_0-1,t_0+1,L^2(\Omega))} \\
\leq & C\Big[\kappa^2 + \kappa^{3/2}\|u\|^\frac 12_{L^2(t_0-1,t_0+1,H^2(\Omega))}\Big]\,.
\endaligned
\end{equation}
Substituting \eqref{eq:86} together with \eqref{eq:84},
and
\eqref{eq:85} into \eqref{eq:81} yields with the aid of \eqref{eq:69}
\begin{equation}\label{eq:80a}
\|u\|_{L^2(t_0,t_0+1,H^2(\Omega))} \leq C\Big[\kappa^3 + \kappa^{3/2}\|u\|^{\frac 12}_{L^2(t_0-1,t_0+1,H^2(\Omega))}\Big]\,
\end{equation}

Proceeding as in the proof of Proposition \ref{lem:asymp-exp}, we can
assume $C\geq 1$ in \eqref{eq:80a} and set
$$
\alpha_n = C^{-1} \kappa^{-\frac 32}  \|u\|_{L^2(n,n+1,H^2(\Omega))} ^\frac 12\,.
$$
We now can  rewrite \eqref{eq:80a} in the form
$$
\alpha_n^2 \leq ( 1 +\alpha_{n-1} + \alpha_n)\,,
$$
which is precisely \eqref{inegaga}.
We can thus conclude (\ref{eq:80}), and hence, for a new value of $C$,
\eqref{eq:87} easily follows.
\vspace{1ex}

{\it Step 2: Prove that
\begin{equation}
  \label{eq:89}
  \|\nabla_{\kappa A_\kappa}u\|_{L^2(t_0,t_0+1,H^1(\Omega,\R^2))}  \leq C \kappa^3 \,.
\end{equation}}
\vspace{2ex}

It can be easily verified that
\begin{equation}
\label{eq:88}
  \|\nabla_{\kappa A_\kappa}u(t,\cdot)\|_{1,2}\leq \|u(t,\cdot)\|_{2,2} + \kappa\|\,|A_\kappa|\nabla u(t,\cdot)\|_2 +
  \kappa\|u\nabla A_{\kappa}(t,\cdot)\|_2 \,.
\end{equation}
Furthermore, in the same manner we have obtained \eqref{eq:86} we
obtain, with the aid of \eqref{eq:87}
\begin{displaymath}
 \kappa  \|\,|A_\kappa|\nabla u\|_{L^2(t_0,t_0+1,L^2(\Omega,\R^2))} \leq C
  \kappa^{3/2}\|u\|_{L^2(t_0,t_0+1,H^2(\Omega,\R^2))}^{1/2}\leq C\kappa^3 \,.
\end{displaymath}
By \eqref{eq:69} and \eqref{eq:72} we have that
\begin{displaymath}
  \kappa\|u\nabla A_{\kappa}\|_{L^2(t_0,t_0+1,L^2(\Omega,\R^2))} \leq C\kappa^2 \,.
\end{displaymath}
We can now conclude \eqref{eq:89} from \eqref{eq:88}. 
\vspace{1ex}

{\it Step 3: Prove \eqref{eq:89a}.}
\vspace{1ex}

 In \cite[(5.35)]{alhe14} it was shown that
 \begin{equation}
\label{eq:90}
    \|\nabla_{\kappa A_\k}\psi_\k\|_{L^\infty(t_0,t_0+1; L^2(\Omega))}\leq C\kappa^3\,.
 \end{equation}
 (Note that while the setting in \cite{alhe14} is different then - in
 particular, we assume there $J\sim\OO(\kappa)$ - the estimate  is
 still valid in the present case because $c=1$.)
Hence, we get
\begin{equation}
\label{eq:91}
  \|\nabla_{\kappa A_\k}u\|_{L^\infty(t_0,t_0+1; L^2(\Omega))}\leq C\kappa^3\,.
\end{equation}
We now use Gagliardo-Nirenberg interpolation inequality (see  \cite{ni59}) to
obtain
$$
\aligned
 &  \|\nabla_{\kappa A_\k}u\|_{L^p(t_0,t_0+1; L^p(\Omega))}^p\qquad  \\
& \leq C \int_{t_0}^{t_0+1}
   \|\nabla_{\kappa A_\k}u(t,\cdot)\|_{1,2}^{p-2}\|\nabla_{\kappa A_\k}u(t,\cdot)\|_2^2 \,dt\\
& \leq C   \|\nabla_{\kappa A_\k}u\|_{L^2(t_0,t_0+1; H^1(\Omega))}^{p-2} \|\nabla_{\kappa
    A_\k}u\|_{L^{\frac{4}{4-p}}(t_0,t_0+1; L^2(\Omega))}^2 \,.
\endaligned
$$
Consequently,
\begin{equation}\label{eq:92}
\aligned
&\|\nabla_{\kappa A_\k}u\|_{L^p(t_0,t_0+1; L^p(\Omega))}^p \\
\leq &C \|\nabla_{\kappa A_\k}u\|_{L^2(t_0,t_0+1; H^1(\Omega))}^{p-2} \|\nabla_{\kappa
    A_\k}u\|_{L^2(t_0,t_0+1; L^2(\Omega))}^{4-p} \|\nabla_{\kappa
    A_\k}u\|_{L^\infty(t_0,t_0+1; L^2(\Omega))}^{p-2} \,.
\endaligned
\end{equation}
 Multiplying (\ref{eq:67}a) by ${\widehat \zeta}^2\bar{\psi}_\kappa$ and integrating over
$\Omega$ we obtain for the real part  that
\begin{displaymath}
  \|\nabla_{\kappa A_\k(t,\cdot)}u(t,\cdot)\|_2^2 \leq \kappa^2\|u (t,\cdot)\|_2^2-\frac{1}{2}
      \frac{d\|u (t,\cdot)\|_2^2}{dt}+ \|\psi_\kappa (t,\cdot) \nabla{\widehat \zeta}\|_2^2 \,,
\end{displaymath}
Integrating over $(t_0,t_0+1)$ and using \eqref{eq:76} we then obtain, for sufficiently large $t_0$,
\begin{displaymath}
   \|\nabla_{\kappa A_\k}u\|_{L^2(t_0,t_0+1; L^2(\Omega))}^2  \leq C \,.
\end{displaymath}
Substituting the above, \eqref{eq:89}, and \eqref{eq:91}, into
\eqref{eq:92} we then obtain \eqref{eq:89a}.
\end{proof}

We can now obtain the following improved regularity for $B_{1,\kappa}$
\begin{proposition}
Let $2<p\leq12/5$. For $0<\delta<\delta_0$,  there exists a constant $C=C(\Omega,\delta)>0$ such that for all $t_0>1$
  and $\kappa>\kappa_0(\delta)$ we have
\begin{equation}
 \label{eq:93}
\|B_{1,\kappa}\|_{L^p(t_0,t_0+1,W^{1,p}(\omega_{\delta,j}))}  \leq C \,.
\end{equation}
\end{proposition}
\begin{proof} 
Taking the curl of \eqref{eq:71} yields that $B_{1,\kappa}$ is a weak
solution of
\begin{equation}
\label{eq:94}
  \begin{cases}
       \frac{\partial  B_{1,\kappa}}{\partial t} - \Delta B_{1,\kappa}  = \frac{1}{\kappa}\curl
       \Im(\bar{\psi}_\k\nabla_{\kappa A_\kappa}\psi_\kappa) & \text{ in } (0,+\infty)\times  \Omega \\
   B_{1,\kappa} =0 & \text{ on } (0,+\infty)\times  \partial\Omega \,.
  \end{cases}
\end{equation}
 Let
\begin{displaymath}
  \B_{\widehat \zeta} = {\widehat \zeta}B_{1,\kappa} \,,
\end{displaymath}
where the cutoff function $\widehat{\zeta}$ is defined by (\ref{eq:77}).
Let further $\hat{\Omega}(\delta)\subset\Omega$ be smooth  and satisfy
$$
{\rm supp}\,\widehat{\zeta}\subset \hat{\Omega}(\delta)\,.$$
As for any $V\in H^1(\Omega,\R^2)$ we have that
\begin{displaymath}
  \curl V = \Div V_\perp\,,
\end{displaymath}
it can be  easily verified
from (\ref{eq:94}) that
\begin{displaymath}
  \frac{\partial  \B_{\widehat \zeta}}{\partial t} - \, \Delta \B_{\widehat \zeta}  =
    \frac{1}{\kappa}{\widehat \zeta}\Div\big( \Im(\bar{\psi}_\k\nabla_{\kappa
      A_\kappa}\psi_\kappa)_\perp\big) -2\Div(B_{1,\kappa}\nabla{\widehat \zeta)}) +
    B_{1,\kappa}\Delta\widehat \zeta\,.
\end{displaymath}
Consequently,
\begin{equation}
\label{eq:95}
\left\{
  \begin{array}{rll}
    \frac{\partial  \B_{\widehat \zeta}}{\partial t} - \, \Delta \B_{\widehat \zeta}  &=
    \frac{1}{\kappa}{\widehat \zeta}\, \Div\,\big( \Im(\bar{\psi}_\k\nabla_{\kappa
      A_\kappa}\psi_\kappa)_\perp\big)& \\
      &\quad  -\Div(2  B_{1,\kappa}\nabla{\widehat \zeta} +
  \nabla \Delta_D^{-1}(B_{1,\kappa}\Delta{\widehat \zeta}) ) & \text{in } (t_0-1,t_0+1)\times\hat{\Omega} \\
   \B_{\widehat \zeta} &=0 & \text{on }  (t_0-1,t_0+1)\times\partial\hat{\Omega} \\
\B_{\widehat \zeta} (t_0-1,\cdot) &  =  \widehat \zeta\, \B_{1,\kappa} (t_0-1,\cdot)  &\text{in } \hat{\Omega}  \,.
  \end{array}
  \right.
\end{equation}
In the above $\Delta_D^{-1}$ denotes the inverse Dirichlet Laplacian in $
\hat{\Omega}$.

In order to apply
  \cite[Theorem 1.6]{by07}  which is devoted to the case of parabolic
  operators written in divergence form and with zero initial condition
  we first decompose the solution of \eqref{eq:95} into two
  Cauchy-Dirichlet problems.
  The first of them is:
  \begin{equation} \label{eq:117a}
\left\{
  \begin{array}{rll}
    \frac{\partial U_1}{\partial t} - \, \Delta U_1 &= \Div f_1
    & \text{ in } (t_0-1,t_0+1)\times\hat{\Omega}\,, \\
   U_1 &=0 & \text{ on }  (t_0-1,t_0+1)\times\partial\hat{\Omega} \,,\\
U_1 (t_0-1,\cdot)& = 0 &\text{ in } \hat{\Omega}  \,,
  \end{array}
  \right.
\end{equation}
in which
\begin{equation}\label{eq:117aa}
f_1 =   \frac{1}{\kappa}{\widehat \zeta}\, \,\big( \Im(\bar{\psi}_\k\nabla_{\kappa
      A_\kappa}\psi_\kappa)_\perp\big)  - 2  B_{1,\kappa}\nabla{\widehat \zeta} -
  \nabla \Delta_D^{-1}(B_{1,\kappa}\Delta{\widehat \zeta}) \,.
 \end{equation}
The second one is:
\begin{equation}
\label{eq:117b}
\left\{
  \begin{array}{rll}
    \frac{\partial U_2}{\partial t} - \, \Delta U_2  &= F_2
 & \text{in } (t_0-1,t_0+1)\times\hat{\Omega} \,,\\
   U_2 &=0 & \text{on }  (t_0-1,t_0+1)\times\partial\hat{\Omega} \,,\\
U_2(t_0-1,\cdot )&  =  \widehat \zeta\, \B_{1,\kappa} (t_0-1,\cdot)  &\text{in } \hat{\Omega}  \,,
  \end{array}
  \right.
\end{equation}
where
 \begin{equation} \label{eq:117ba}
 F_2:= -  \frac{1}{\kappa}  \,\big( \Im(\bar{\psi}_\k \nabla {\widehat \zeta}\cdot (\nabla_{\kappa
      A_\kappa}\psi_\kappa)_\perp) \big)\,.
 \end{equation}
By uniqueness of the weak solution, we have
\begin{equation} \label{eq:117c}
 \B_{\widehat \zeta}  = U_1 + U_2  \text{ in } (t_0-1,t_0+1)\times\hat{\Omega}\,.
 \end{equation}
We now separately estimate $U_1$ and $U_2$ in  $L^p (t_0,t_0+1, W^{1,p}( \hat \Omega))\,$.
\vskip0.05in

{\it Estimate of $U_1$}
\vspace{1ex}

We apply  \cite[Theorem 1.6]{by07} to obtain
\begin{equation} \label{eq:96}
\| U_1\|_{L^p(t_0,t_0+1,W^{1,p}(\Omega))}\\
\leq C \,\|f_1\|_{L^p(t_0,t_0+1,L^p(\Omega,\R^2)}\,.
\end{equation}
It can be easily verified, by the Gagliardo-Nirenberg interpolation
inequality \cite{ni59} that, for
all $2<p<4$, there exists a constant $C$ such that,
for any $\phi \in L^2(t_0,t_0+1,H^1(\Omega))\cap L^\infty (t_0,t_0+1,L^2(\Omega))$, we have \begin{multline}\label{eq:76a}
  \|\phi \|_{L^p(t_0,t_0+1,L^p(\Omega,\R^2))}^p\leq C \int_{t_0}^{t_0+1}\|\phi (\tau,\cdot)
  \|_2^2 \|\phi (\tau,\cdot) \|_{1,2}^{p-2}\,d\tau \\ \leq C  \| \phi \|_{L^{\frac{4}{4-p}}(t_0,t_0+1,L^2(\Omega,\R^2))}^2\|\phi\|_{L^2(t_0,t_0+1,H^1(\Omega,\R^2))}^{p-2}\,.
\end{multline}
By \eqref{eq:76a},  \eqref{eq:72}, Remark~\ref{usefulrem}, and Sobolev embeddings we
have:
\begin{equation}\label{eq:97}
\| \B_{\widehat \zeta} (t_0,\cdot)\|_p +
\|\B_{1,\kappa} \nabla{\widehat \zeta}\|_{L^p(t_0,t_0+1,L^p(\Omega))} +
\|\B_{1,\kappa}
\Delta{\widehat \zeta}\|_{W^{-1,p}(t_0,t_0+1,L^p(\Omega,\R^2))}
\leq C\,.
\end{equation}
Furthermore, by \eqref{eq:69} we have that
\begin{displaymath}
  \begin{array}{l}
\|{\widehat \zeta}\Im(\bar{\psi}_\kappa\nabla_{\kappa A_\kappa}\psi_\kappa)\|_{L^p(t_0,t_0+1,L^p(\Omega,\R^2))}\\
\qquad   \leq
\|\psi_\kappa\nabla{\widehat \zeta}\|_{L^p(t_0,t_0+1,L^p(\Omega,\R^2))}+  \|\nabla_{\kappa
  A_\kappa}({\widehat \zeta}\psi_\kappa)\|_{L^p(t_0,t_0+1,L^p(\Omega,\R^2))} \\
\qquad \leq \|\nabla_{\kappa A_\kappa}({\widehat \zeta}\psi_\kappa)\|_{L^p(t_0,t_0+1,L^p(\Omega,\R^2))} +
C\,.
\end{array}
\end{displaymath}
Substituting the above together with \eqref{eq:97}
into \eqref{eq:96} yields for $t_0$ large enough
\begin{equation}\label{eq:98}
\| U_1\|_{L^p(t_0,t_0+1,W^{1,p}(\Omega))} \leq C\Big(1 +\frac{1}{\kappa}
 \|\nabla_{\kappa A_\kappa}({\widehat \zeta}\psi_\kappa)\|_{L^p(t_0,t_0+1,L^p(\Omega,\R^2))}\Big)  \,.
\end{equation}
 Substituting \eqref{eq:89a}
 into \eqref{eq:98} yields
\begin{equation}\label{eq:79a}
\| U_1\|_{L^p(t_0,t_0+1,W^{1,p}(\Omega))}
\leq C(1+\kappa^{5-12/p})
 \leq  \widehat C\,,
\end{equation}
since $2< p\leq 12/5$.
\vspace*{2ex}

{\it Estimate of $U_2$}
\vspace{1ex}

Here we apply first $L^2$ estimates and then combine them with Sobolev's
estimates. It is in this part that we need the information on $F_2$ and $U_2$
in $[t_0-1,t_0+1)\times \hat \Omega$ in order to bound the various norms on
$(t_0,t_0+1)\times \hat \Omega$. We begin by applying once again \cite[Theorem
C.1]{alhe14} (combining $(C.1)$ and $(C.2)$ there) to obtain
 $$
 \begin{array}{l}
 \| U_2\|_{ L^2(t_0,t_0+1, H^2)} + \| U_2 \|_{L^\infty (t_0,t_0+1,H^1)}\\ \qquad
  \leq C \left(  \| F_2\|_{L^2 ((t_0-1,t_0 +1)\times \hat \Omega)} + \| U_2(t_0-1,\cdot)\|_{L^2(\hat \Omega)} \right)\,,
  \end{array}
  $$
  where $F_2$ and $U_2(t_0-1,\cdot)$ given in \eqref{eq:117ba} and
  \eqref{eq:117b}.  Applying Gagliardo-Nirenberg's inequality yields,
  for $ 2 <p <4$,
  \begin{equation}\label{eq:89b}
   \| U_2\|_{ L^p(t_0,t_0+1,W^{1,p})} \leq C \left(  \| F_2\|_{L^2
       ((t_0-1,t_0 +1)\times \hat \Omega)} + \| U_2(t_0-1,\cdot)\|_{L^2(\hat \Omega)}
   \right)\,,
  \end{equation}
By \eqref{eq:70aa} we have that
\begin{displaymath}
  \|  F_2\|_{L^2((t_0-1,t_0 +1)\times \hat \Omega)} \leq C \,.
\end{displaymath}
Furthermore, using \eqref{eq:72}, with Remark \ref{usefulrem} in mind
yields
\begin{displaymath}
  \| U_2(t_0-1,\cdot)\|_2 \leq C \,.
\end{displaymath}
Consequently, by \eqref{eq:89b}, there exist, for any $2 <p <4$ and
any $\delta >0\,,$ constants $C(\delta)$ and $\kappa(\delta)$ such that for any $\kappa
\geq \kappa_0(\delta)$ and any $t_0 >1$ we have
   \begin{equation}
\label{eq:89c}
   \| U_2\|_{ L^p(t_0,t_0+1,W^{1,p})} \leq C (\delta)\,.
  \end{equation}
 The combination of \eqref{eq:79a} and \eqref{eq:89c}  together with
 \eqref{eq:117c} completes the proof of the proposition.
\end{proof}

We can now establish the exponential decay of $\psi_\kappa$.
\begin{proposition}
  \label{asymp-exp-2}
  Let $\omega_{\delta,j}$ ($j\in\{1,2\}$) be given by \eqref{eq:56}.  Suppose
  that for some $k\in\{1,2\}$ \eqref{eq:75} is satisfied.  Then, there
  exist $C>0$ and $\delta_0>0$, and, for any  $0<\delta<\delta_0$,
  $\kappa_0(\delta)$, such that for any  $\k\geq \k_0(\delta)$ we have
\begin{equation}
\label{eq:99}
\limsup_{t\to\infty}\int_{\omega_{\delta,j}} \exp \Big(\delta^{1/2}
 \kappa d(x,
\Gamma_{\delta,j})\Big) |\psi_\kappa|^2(t,x) \,dx \leq C  \,.
\end{equation}
\end{proposition}
\begin{proof} 
  Without loss of generality we may, as before,  assume $h_j>0$.
Let $\check{\chi}$ and $\check{\zeta}$ be defined by \eqref{check1} and
\eqref{check2}.
\vspace{1ex}
 
{\it Step 1: Prove that
\begin{equation}
\label{eq:101}
  \|\check{\zeta}\psi_\kappa(t,\cdot)\|_2^2 \leq \|\check{\zeta}\psi_0\|_2^2e^{-2\gamma\kappa^2t} +
  \frac{C(\delta)}{\kappa^4}+ \int_0^t
  e^{-2\gamma\kappa^2(t-\tau)}\big|\langle\kappa B_{1,\kappa}\check{\zeta}\psi_\kappa,\check{\zeta}\psi_\kappa\rangle(\tau)\big|\,d\tau \,.
\end{equation}}
\vspace{1ex}
      
Multiplying
(\ref{eq:67}a) by $\check{\zeta}^2 \bar{\psi}$ and integrating by parts yields
\begin{displaymath}
  \frac{1}{2}\frac{d}{dt}\left( \|\check{\zeta}\psi_\kappa(t,\cdot) \|_2^2 \right)  +  \|\nabla_{\kappa A_\kappa(t,\cdot)}(\check{\zeta}\psi_\kappa(t,\cdot))\|_2^2 \leq
  \kappa^2\|\check{\zeta}\psi_\kappa(t,\cdot)\|_2^2 + \|\psi_\kappa (t,\cdot) \nabla\check{\zeta}\|_2^2 \,.
\end{displaymath}
By Theorem 2.9 in \cite{AHS} we have
$$\aligned
  \|\nabla_{\kappa A_\kappa}(\check{\zeta}\psi_\kappa)\|_2^2
&  \geq \langle\kappa B_\kappa(t,\cdot) \check{\zeta}\psi_\kappa(t,\cdot),\check{\zeta}\psi_\kappa (t,\cdot)\rangle\\  & =
  \kappa^2\langle B_n\check{\zeta}\psi_\kappa,\check{\zeta}\psi_\kappa\rangle + \langle\kappa B_{1,\kappa}\check{\zeta}\psi_\kappa,\check{\zeta}\psi_\kappa\rangle  \\
&\geq
  \kappa^2\Big(1+\frac{\delta}{2}\Big)\|\check{\zeta}\psi_\kappa\|_2^2+ \langle\kappa B_{1,\kappa}\check{\zeta}\psi_\kappa,\check{\zeta}\psi_\kappa\rangle \,.
\endaligned
$$
We can thus write
\begin{equation}\label{eq3.32}
\aligned
&\frac{1}{2}\frac{d}{dt}\left( \|\check{\zeta}\psi_\kappa(t,\cdot)\|_2^2\right)  +
\kappa^2\Big( \frac{\delta}{2}- \frac{\delta}{4}
\Big)\|\check{\zeta}\psi_\kappa(t,\cdot)\|_2^2\\
\leq&
  \|\psi_\kappa (t,\cdot) \nabla\eta\|_2^2 + \|\tilde \chi\,\psi_\kappa (t,\cdot)\nabla\eta_\delta\|_2^2- \langle\kappa B_{1,\kappa}(t,\cdot)\check{\zeta}\psi_\kappa(t,\cdot),\check{\zeta}\psi_\kappa(t,\cdot)\rangle \,.
\endaligned
\end{equation}
By \eqref{eq:76}, for every $0<\delta\leq\delta_0$, we have
\begin{displaymath}
  \|\psi_\kappa (t,\cdot) \nabla\eta\|_2^2 + \|\tilde \chi\, \psi_\kappa (t,\cdot) \,\nabla\eta_\delta\|_2^2 \leq \frac{C}{\kappa^2} \,,
\end{displaymath}
which when substituted into \eqref{eq3.32} yields
\begin{equation}
 \label{eq:100}
   \frac{1}{2}\frac{d}{dt}\left( \|\check{\zeta}\psi_\kappa (t,\cdot)\|_2^2\right)  +   \kappa^2\gamma\|\check{\zeta}\psi_\kappa (t,\cdot) \|_2^2
   \leq  \frac{C}{\kappa^2} - \langle\kappa B_{1,\kappa}(t,\cdot) \check{\zeta}\psi_\kappa(t,\cdot) ,\check{\zeta}\psi_\kappa (t,\cdot)\rangle \,,
\end{equation}
where $\gamma=\delta/4$. We now get \eqref{eq:101} from \eqref{eq:100}. 
\vspace{1ex}

{\it Step 2: Prove that,  for all $n\geq2$,
\begin{equation}
\label{eq3.33}
\aligned
&\|\check{\zeta}\psi_\kappa(\tau, \cdot )\|_{L^\infty(t^*+n-1,t^*+n,L^2(\Omega))}^2 \\
\leq& C \kappa^{-(p-2)/p}\|\check{\zeta}\psi_\kappa(\tau, \cdot )\|_{L^\infty(t^*+n-2,t^*+n,L^2(\Omega))}^2  +{C\over\k^4} \,,
\endaligned
\end{equation}
where $C$ is independent of $\kappa\,$.}
\vspace{1ex}

To prove \eqref{eq3.33} we need to estimate the last term on the
right-hand-side of \eqref{eq:101}. To this end we first write
\begin{multline*}
  \int_0^t  e^{-2\gamma\kappa^2(t-\tau)}\big|\langle\kappa
  B_{1,\kappa}\check{\zeta}\psi_\kappa,\check{\zeta}\psi_\kappa\rangle (\tau) \big|\,d\tau  \\
\leq \kappa \left(\int_{t-1}^t  e^{-2\gamma\kappa^2(t-\tau)}
{  \|B_{1,\kappa}(\tau, \cdot
  )\|_{L^\infty(\omega_{\delta/2,j})}}\,d\tau\right)\,\cdot\|\check{\zeta}\psi_\kappa\|_{L^\infty(t-1,t;L^2(\Omega))}\\
+ \kappa\int_0^{t-1}
e^{-2\gamma\kappa^2(t-\tau)}\|B_{1,\kappa}(\tau,\cdot)\|_1\,d\tau
\|\check{\zeta}\psi_\kappa\|_{L^\infty(0,t-1;L^2(\Omega))}\,.
\end{multline*}
For the last term on the right-hand-side we have in view of Remark~\ref{usefulrem},  \eqref{check1}, \eqref{check2}, and \eqref{eq:69} that
\begin{displaymath}
  \kappa\int_0^{t-1}
e^{-2\gamma\kappa^2(t-\tau)}\|B_{1,\kappa}(\tau,\cdot)\|_1\,d\tau
\|\check{\zeta}\psi_\kappa\|_{L^\infty(0,t-1;L^2(\Omega))}\leq Ce^{-2\gamma\kappa^2}e^{C\kappa}\,.
\end{displaymath}
Hence for sufficiently large $\kappa$ we have
\begin{multline}
\label{eq:102}
  \int_0^t  e^{-2\gamma\kappa^2(t-\tau)}\big|\langle\kappa
  B_{1,\kappa}\check{\zeta}\psi_\kappa,\check{\zeta}\psi_\kappa\rangle (\tau) \big|\,d\tau \leq \\
\kappa \left(\int_{t-1}^t  e^{-2\gamma\kappa^2(t-\tau)}  \|B_{1,\kappa}(\tau, \cdot
  )\|_{L^\infty(\omega_{\delta/2,j})}\,d\tau\right)\,\cdot\|\check{\zeta}\psi_\kappa\|_{L^\infty(t-1,t;L^2(\Omega))}+ Ce^{-\gamma\kappa^2}\,.
\end{multline}
Since by Sobolev embeddings
\begin{displaymath}
  \|B_{1,\kappa}\|_{L^p(t-1,t,L^\infty(\omega_{\delta/2,j}))} \leq C\|B_{1,\kappa}\|_{L^p(t-1,t,W^{1,p}(\omega_{\delta/2,j}))} \,,
\end{displaymath}
we can use \eqref{eq:93}  to obtain, for
sufficiently large $t$ and $\kappa$, and for any $2<p\leq12/5\,$,
$$ \aligned
& \int_{t-1}^t  e^{-2\gamma\kappa^2(t-\tau)}\|B_{1,\kappa}(\tau, \cdot)\|_{L^\infty(\omega_{\delta/2,j})}\,d\tau \\
\leq&
 \|B_{1,\kappa}(\tau, \cdot)\|_{L^p(t-1,t;L^\infty(\omega_{\delta/2,j}))}\,\Big[\int_{t-1}^t
 e^{-\frac{2p}{p-1}\gamma\kappa^2(t-\tau)}\,d\tau\Big]^{\frac{p-1}{p}} \\
\leq&
 C\kappa^{-2(p-1)/p}\,.
\endaligned
$$
Substituting the above into \eqref{eq:102} yields
$$\aligned
&\int_0^t  e^{-2\gamma\kappa^2(t-\tau)}\big|\langle\kappa B_{1,\kappa}\check{\zeta}\psi_\kappa,\check{\zeta}\psi_\kappa\rangle\big|\,d\tau \\ 
 \leq& C\kappa^{-(p-2)/p}\|\check{\zeta}\psi_\kappa\|_{L^\infty(t-1,t;L^2(\Omega))} + Ce^{- \frac 14 \delta \kappa^2 } \,,
\endaligned
$$
which, when substituted into \eqref{eq:101},  yields \eqref{eq3.33}. 
\\

{\it Step 3: Prove \eqref{eq:99}.}
\\

Let now
\begin{displaymath}
  b_n = \|\check{\zeta}\psi_\kappa(\tau, \cdot )\|_{L^\infty(t^*+n-1,t^*+n,L^2(\Omega))}^2 \,.
\end{displaymath}
From \eqref{eq3.33} we get that if $p=12/5$ then for sufficiently large
$\kappa$  it holds that
\begin{displaymath}
b_n \leq C \left( \kappa^{-1/6}(b_{n-1} + b_n)+  1\right)\,,
\end{displaymath}
where $C$ is independent of $\kappa$.
For another  constant $\widehat C$, we get for sufficiently large $\kappa$,
\begin{displaymath}
b_n \leq \widehat C (\kappa^{-1/6} b_{n-1} +  1)\,.
\end{displaymath}
This immediately implies, for $\kappa$ large enough so that $\widehat C
\kappa^{-\frac 16} \leq \frac 12$, the upperbound
\begin{displaymath}
  \limsup_{n\to\infty}b_n \leq C_0\,,
\end{displaymath}
where $C_0$ is independent of $\kappa$. Consequently,
$$
\limsup_{t\to\infty}\int_{\omega_{\delta,j}} \exp\Big(\delta^\frac 12 \kappa d(x,
\partial\omega_{\delta,j})\Big) |\psi_\kappa(t,x)|^2 \,dx \leq C_0\,,
$$
which readily yields \eqref{eq:99}.
\end{proof}
\newpage

\section{Large domains}
\label{sec:5}
The main goal of this section is to prove Proposition
\ref{largedomain}. To this end it is more convenient to consider \eqref{eq:18}
in a fixed domain. Assuming that $0\in\O$ we set $\epsilon=1/R$ (hence we have $\epsilon \ll 1$)
and apply the transformation,
\begin{equation}
\label{eq:103}
\psi_\epsilon(\epsilon x) = \psi(x)   \, , \qquad A(x) = \epsilon^{-1}A_\epsilon(\epsilon x)\,, \qquad
  \phi(x) = \phi_\epsilon(\epsilon x) \,,
\end{equation}
If we write $y=\epsilon x$, we have:
  \begin{displaymath}
     \curl_x^2A  =  \epsilon\,\curl_y^2A_\epsilon \quad ; \quad
     \nabla_x\phi = \epsilon\,\nabla_y\phi_\epsilon  \quad ;
     \quad  \nabla_A\psi  =  \epsilon\,\nabla_{\epsilon^{-2}A_\epsilon}\psi_\epsilon \,,
  \end{displaymath}
  leading to the following system for $(\psi_\epsilon, A_\epsilon, \phi_\epsilon)$
\begin{subequations}
\label{eq:104}
\begin{alignat}{2}
& \Delta_{\epsilon^{-2}A_\epsilon}\psi_\epsilon + \frac{\psi_\epsilon}{\epsilon^2}   \big( 1 - |\psi_\epsilon|^{2}
\big)-\frac{i}{\epsilon^2}\phi_\epsilon\psi_\epsilon =0 & \quad \text{ in } \Omega\, ,\\
 & \curl^2A_\epsilon +  \nabla\phi_\epsilon  =  \Im\big(\bar\psi_\epsilon
 \nabla_{\epsilon^{-2}A_\epsilon}\psi_\epsilon\big)  & \quad \text{ in }  \Omega\,,\\
  &\psi_\epsilon=0 &\quad \text{ on }  \partial\Omega_c \,, \\
 & \nabla_{\epsilon^{-2}A_\epsilon}\psi_\epsilon\cdot\nu=0 & \quad \text{ on }  \partial\Omega_i \,,\\
 & \frac{\partial\phi_\epsilon}{\partial\nu}= f(\epsilon)J  &\quad \text{ on } \partial\Omega \,,\\
& \frac{\partial\phi_\epsilon}{\partial\nu}= 0 &\quad \text{ on } \partial\Omega_i\,,\\[1.2ex]
&\dashint_{\partial\Omega}\curl A_\epsilon(x)\,ds = f(\epsilon)h_{ex}\,.
\end{alignat}
\end{subequations}
In the above
$$
f(\epsilon)=F(1/\epsilon)= \epsilon^{-\alpha}\,.
$$
It follows from \eqref{eq:6} that
\begin{displaymath}
  h_j = b_j\epsilon^{-\alpha}\,, \quad j=1,2\,,
\end{displaymath}
where $b_j$ is independent of $\epsilon$ for $j=1,2\,$.\\

We assume that $A_\epsilon$ is in the Coulomb gauge space \eqref{eq:5}, and
suppose that a weak solution $(\psi_\epsilon,A_\epsilon,\phi_\epsilon)\in H^1(\Omega,\C)\times
H^1(\Omega,\R^2)\times L^2(\Omega)$ exists.  Proposition \ref{largedomain} can now be
reformulated in the following way:

\begin{proposition}
\label{semiclassical}
 Let $(\psi_\epsilon,A_\epsilon,\phi_\epsilon)$ denote a solution of \eqref{eq:104}, and let
$h$ be given by \eqref{eq:9}.
  Suppose that for some $0<\gamma<1$ and $\epsilon_0>0$  we have
  \begin{displaymath}
    \epsilon^{-\gamma}<h\,, \quad \forall\, 0<\epsilon<\epsilon_0\,.
  \end{displaymath}
Then,  there exists a compact  set $K\subset\Omega$, $C>0$, and $\alpha >0$,
such that for any  $0\,<\,\epsilon\,<\epsilon_0\,$ we have
\begin{equation}
\label{eq:105}
\int_{K} |\psi_\epsilon(x)|^2 \,dx \leq C  e^{-\alpha/\epsilon} \,.
\end{equation}
\end{proposition}

We split the proof of Proposition \ref{semiclassical} into several
steps, to each of them we dedicate a separate lemma. We begin by observing,
as in Section \ref{sec:3}, that
\begin{equation}
  \label{eq:106}
\|\psi_\epsilon\|_\infty\leq 1\,.
\end{equation}
Let
\begin{equation}
\label{eq:107}
  A_{1,\epsilon}=A_\epsilon-\epsilon^{-\alpha}A_n,\q
\phi_{1,\epsilon}=\phi_\epsilon-\epsilon^{-\alpha}\phi_n,
\end{equation}
Set further
\begin{equation}
\label{eq:108}
  B_{1,\epsilon}= \curl A_{1,\epsilon}\,;  \quad B_\epsilon= \curl A_\epsilon  \,.
\end{equation}
By (\ref{eq:104}b) and (\ref{eq:11}a),   we then have
\begin{subequations}
\label{eq:109}
\begin{empheq}[left={\empheqlbrace}]{alignat=2}
  &\nabla_\perp B_{1,\epsilon} +  \nabla\phi_{1,\epsilon}  =  \Im(\bar\psi_\epsilon\, \nabla_{\epsilon^{-2}A_\epsilon}\psi_\epsilon)  &  \text{ in }
  \Omega\,,\\
   &\frac{\partial\phi_{1,\epsilon}}{\partial\nu}= 0 & \text{ on } \partial\Omega \,, \\
&\dashint_{\partial\Omega}B_{1,\epsilon}(x)\,ds = 0 \,. &
\end{empheq}
\end{subequations}
Note that since $\partial B_{1,\epsilon}/\partial\tau=\partial\phi_{1,\epsilon}/\partial\nu=0$ on $\partial\Omega$
we must have by (\ref{eq:109}c) that
\begin{equation}
\label{eq:110}
  B_{1,\epsilon}|_{\partial\Omega}\equiv 0\,.
\end{equation}

 We begin with the following auxiliary estimate.
\begin{lemma}
  \label{lem:magnetic-bound}
Let $w_\epsilon$ denote the solution of
\begin{equation}\label{eq:111}
 \left\{\aligned
&    \Delta w_\epsilon - \frac{1}{\epsilon^2}|\psi_\epsilon|^2w_\epsilon = 0\q & \text{\rm in }\Omega\;, \\
&    w_\epsilon = B_\epsilon - 1 \q& \text{\rm on } \partial\Omega.
\endaligned
\right.
\end{equation}
Under the assumptions on $J$ and $\Omega$ in (\ref{eq:3})-(\ref{hyptopolo}) we have
\begin{equation}
 \label{eq:112}
\|B_\epsilon-1-w_\epsilon\|_\infty \leq  \frac{1}{2} \,.
\end{equation}
Furthermore, we have
\begin{equation}
\label{eq:113}
\|\nabla\phi_\epsilon\|_2+ \|\phi_\epsilon\|_\infty \leq  C\epsilon^{-\alpha}   \,.
\end{equation}
\end{lemma}
\begin{proof}
  As can be easily verified from (\ref{eq:104}a) we have (see in
  \cite{al02} (formula (2.4) in the case when $\phi_\epsilon=0$),
  \begin{equation}
    \label{eq:114}
 \frac 12 \Delta|\psi_\epsilon|^2 = -\frac{|\psi_\epsilon|^2}{\epsilon^2}(1-|\psi_\epsilon|^2) +
|\nabla_{\epsilon^{-2}A_\epsilon}\psi_\epsilon|^2 \,.
  \end{equation}
Furthermore, taking the curl of (\ref{eq:104}b) yields
\begin{displaymath}
  \Delta B_\epsilon - \frac{1}{\epsilon^2}|\psi_\epsilon|^2B_\epsilon  =
   -\Im (\nabla_{\epsilon^{-2}A_\epsilon}\bar\psi_\epsilon\times\nabla_{\epsilon^{-2}A_\epsilon} \psi_\epsilon) \,.
\end{displaymath}
Note that
$$\aligned
&\curl \Im(\bar\psi_\epsilon\, \nabla_{\epsilon^{-2}A_\epsilon}\psi_\epsilon) \\
=&\Im(\nabla\bar\psi_\epsilon\times \nabla_{\epsilon^{-2}A_\epsilon}\psi_\epsilon
  -i\epsilon^{-2}\bar\psi_\epsilon A_\epsilon\times\nabla\psi_\epsilon) -  \frac{1}{\epsilon^2}|\psi_\epsilon|^2B_\epsilon\\
=&\Im (\nabla_{\epsilon^{-2}A_\epsilon}\bar \psi_\epsilon\times\nabla_{\epsilon^{-2}A_\epsilon} \psi_\epsilon) -
  \frac{1}{\epsilon^2}|\psi_\epsilon|^2B_\epsilon \,.
\endaligned
$$
Let 
$$u_\epsilon=B_\epsilon-1 + {|\psi_\epsilon|^2\over 2}-w_\epsilon\,.
$$
Combining the above and \eqref{eq:114}
yields that (cf. also \cite{al02})
\begin{displaymath}
  \begin{cases}
    \Delta u_\epsilon - \frac{1}{\epsilon^2}|\psi_\epsilon|^2u_\epsilon = |\nabla_{\epsilon^{-2}A_\epsilon}\psi_\epsilon|^2 -\Im (\nabla_{\epsilon^{-2}A_\epsilon}\bar \psi_\epsilon\times\nabla_{\epsilon^{-2}A_\epsilon} \psi_\epsilon)
 + \frac{1}{2\epsilon^2} |\psi_\epsilon|^4\geq 0 &
\text{in } \Omega, \\
u_\epsilon = \frac{|\psi_\epsilon|^2}{2} & \text{on } \partial\Omega.
  \end{cases}
\end{displaymath}
By the weak maximum principle (cf. for instance \cite[Theorem
8.1]{GT}) and \eqref{eq:106} we obtain that for sufficiently small $\epsilon$
\begin{displaymath}
  u_\epsilon (x)\leq \frac{1}{2}\,.
\end{displaymath}
The lower bound in \eqref{eq:112} follows easily by setting
$$\tilde{u}_\epsilon = -B_\epsilon-1 + {|\psi_\epsilon|^2\over 2}+w_\epsilon
$$ 
to obtain
\begin{displaymath}
  \begin{cases}
    \Delta \tilde{u}_\epsilon - \frac{1}{\epsilon^2}|\psi_\epsilon|^2\tilde{u}_\epsilon = |\nabla_{\epsilon^{-2}A_\epsilon}\psi_\epsilon|^2  +
\Im (\nabla_{\epsilon^{-2}A_\epsilon}\bar \psi_\epsilon\times\nabla_{\epsilon^{-2}A_\epsilon} \psi_\epsilon)  + \frac{1}{2\epsilon^2} |\psi_\epsilon|^4 \geq 0 &
\text{in } \Omega, \\
\tilde{u}_\epsilon = \frac{|\psi_\epsilon|^2}{2} & \text{on } \partial\Omega,
  \end{cases}
\end{displaymath}
upon which we use again the weak maximum principle.

To prove \eqref{eq:113} we first obtain for $\phi_\epsilon$, in the same manner
used to derive \eqref{eq:33}, the following problem:
\begin{equation}\label{eq:115}
\left\{
\aligned
-&\Delta\phi_\epsilon + \frac{\rho^2_\epsilon}{\epsilon^2}\phi_\epsilon =0 \q& \text{\rm in } \Omega\,, \\
&  \frac{\partial\phi_\epsilon}{\partial\nu}= \epsilon^{-\alpha}J \q&  \text{\rm on } \partial\Omega_c\,,\\
& \frac{\partial\phi_\epsilon}{\partial\nu}= 0 \q&  \text{\rm on } \partial\Omega_i\,.
\endaligned\right.
\end{equation}
Then, we follow the same steps as in the proof of \eqref{eq:34} to
obtain the bound $\|\phi_\epsilon\|_\infty\leq C \epsilon^{-\alpha}$. Multiplying \eqref{eq:115} by $\phi_\epsilon$ and
integrating by parts yields, using the preceding $L^\infty$ bound \eqref{eq:112}, we find that
\begin{displaymath}
  \|\nabla\phi_\epsilon\|_2^2 \leq \int_{\partial\Omega} \phi_\epsilon  \frac{\partial\phi_\epsilon}{\partial\nu}\,ds \leq C\epsilon^{-2\alpha}\,.
\end{displaymath}
\end{proof}
As a corollary we get:
\begin{corollary}
  \begin{equation}
\label{eq:116}
\|B_\epsilon\|_\infty \leq  \max(|b_1|,b_2)\epsilon^{-\alpha} + \frac{1}{2} \,.
\end{equation}
\end{corollary}
\begin{proof}
  By \eqref{eq:112} and the maximum principle we have that
  \begin{displaymath}
    \|B_\epsilon-1\|_\infty\leq \|w_\epsilon\|_\infty +\frac{1}{2} \leq  \max(|b_1|,b_2)\epsilon^{-\alpha} - \frac{1}{2}\,,
  \end{displaymath}
which readily yields \eqref{eq:116}.
\end{proof}

We continue with the following auxiliary estimate:
\begin{lemma}\label{Lem3}
  Let $\Omega$ and $J$ satisfy (\ref{eq:3})-(\ref{hyptopolo}), and $w_\epsilon$ be a solution of \eqref{eq:111}.
There exist positive constants $C$ and $\epsilon_0$, such that, for all $0<\epsilon<\epsilon_0$,
  \begin{equation}
\label{eq:117}
\|\nabla w_\epsilon \|_\infty \leq
\frac{C}{\epsilon^{1+\alpha}} \,.
  \end{equation}
\end{lemma}
\begin{proof}
  For convenience of notation we drop the subscript $\epsilon$ in the proof
  and bring only its main steps, as it rather standard. We first apply
  the inverse transformation of \eqref{eq:103} to \eqref{eq:111} to
  obtain
\begin{displaymath}
   \begin{cases}
    \Delta w - |\psi|^2w = 0 & \text{in }\Omega_R\,, \\
    w = B - 1 & \text{on } \partial\Omega_R \,,
  \end{cases}
\end{displaymath}
where $B=\curl A$.\\
We distinguish in the following between interior estimates and
boundary estimates.  Let $x_0\in \partial \Omega_R$ and
$D_r=D_r(x_0)=B(x_0,r)\cap\Omega_R\,.
$\\
By the standard elliptic estimates we then have, in view of
\eqref{eq:106},
\begin{equation}
\label{eq:124}
  \|w\|_{H^2(D_r)} \leq C(\|w\|_{L^2(D_{2r})}+ \| \epsilon^{-\alpha}B_n-1\|_{H^2(D_{2r})}) \,.
\end{equation}
To obtain the above we first observe that $B = \epsilon^{-\alpha} B_n$ on the
boundary, and then use the fact that the trace of $B_n$ in $H^\frac
32 (\partial\Omega)$ and is therefore bounded from above by a proper $H^2$ norm.

Similarly,
\begin{equation}
\label{eq:125}
  \|w\|_{H^3(D_r)} \leq C\Big(\|w\|_{H^1(D_{2r})} + \|\epsilon^{-\alpha}B_n-1\|_{H^3(D_{2r})}
 +\|w\|_{L^\infty(D_{2r})} \||\psi|^2\|_{H^1(D_{2r})}\Big) \,.
\end{equation}
Using Kato's inequality and  (\ref{eq:18}a) yields
\begin{displaymath}
  \|\nabla|\psi|\,\|_{L^2(D_r)} \leq \|\nabla_A\psi\|_{L^2(D_r)} \leq C\|\psi\|_{L^2(D_{2r})}\,,
\end{displaymath}
then we obtain by \eqref{eq:112} that
\begin{equation}\label{112aa}
  \|w\|_{H^3(D_r)} \leq C\epsilon^{-\alpha}\,.
\end{equation}

An interior estimate is even easier.  Consider $x_0\in \Omega_R$ such that
$D(x_0,2r) \subset \Omega_R$. There is no need in this case to include a
boundary term in \eqref{eq:124} and \eqref{eq:125}. We then obtain
\eqref{112aa} in this case in a similar manner.

The proof of \eqref{eq:117} now follows from Sobolev embeddings and
\eqref{eq:103}.
\end{proof}

We next define the following subdomain of $\Omega$:
 \begin{displaymath}
\Dg_\delta(\epsilon) = \{ x\in\Omega \,:\,  |B_\epsilon(x)|<\delta\epsilon^{-\alpha} \} \,.
 \end{displaymath}
 Let further
 \begin{displaymath}
   d_{\delta,j}(\epsilon) = d( \Dg_\delta(\epsilon) ,\partial\Omega_{i,j}),\quad j=1,2\,,
 \end{displaymath}
where, as in the introduction,  $\{\partial\Omega_{i,j}\}_{j=1}^{2}$ denotes the set of connected components of
$\partial\Omega_i$. We now obtain a lower bound of
\eq\label{d-delta}
d_\delta(\epsilon)=\max_{j\in\{1,2\}}d_{\delta,j}(\epsilon) .
\eeq
\begin{lemma}
  Let
  \begin{equation}
    \label{eq:118}
\delta_0=\min(|b_1|,|b_2|)\,.
  \end{equation}
 Under the conditions of Lemma \ref{Lem3}, there exists, for any   $0<\delta<\delta_0\,$, a
  positive  $C_\delta$ such that for sufficiently small $\epsilon$ we have
  \begin{equation}
\label{eq:119}
d_\delta(\epsilon) \geq  C_\delta\, \epsilon  \,.
  \end{equation}
\end{lemma}
\begin{proof}
  Let $x\in\partial\Omega_i$ and $y\in\Dg_\delta(\epsilon)$.  By \eqref{eq:117} we have
\begin{equation}
\label{eq:120}
 |w_\epsilon(x)-w_\epsilon(y)| \leq   \frac{C}{\epsilon^{1+\alpha}} |x-y|\,.
\end{equation}
By \eqref{eq:112} we have
\begin{displaymath}
  |w_\epsilon(x)-w_\epsilon(y)| \geq |B_\epsilon(x)-B_\epsilon(y)| -\frac{1}{2}\geq  (\delta_0-\delta)\epsilon^{-\alpha}-\frac{1}{2}\,.
\end{displaymath}
Combining the above with \eqref{eq:120} yields
\begin{displaymath}
   |x-y| \geq (\delta_0-\delta)\epsilon-\frac{1}{2}\epsilon^{1+\alpha} \,,
\end{displaymath}
which readily yields \eqref{eq:119}.
\end{proof}

Next we show
\begin{lemma} Under the conditions of Lemma \ref{Lem3}, there exist $C>0$, $\epsilon_0>0$ and $\delta_0>0$ such that,  for
   $0<\epsilon\leq \epsilon_0 $ and $0< \delta \leq \delta_0\,$,
  \begin{equation}
\label{eq:121}
\int_{\Omega\setminus\Dg_\delta(\epsilon)} \exp \left( [2\delta \Theta_0\epsilon^{-\alpha}]^{1/2}(4\epsilon)^{-1}d(x,\Dg_\delta(\epsilon))\right)  |\psi_\epsilon|^2 \,dx \leq \frac{C}{\delta^{3/2}} \,.
  \end{equation}
\end{lemma}
\begin{proof}
  Let $\eta\in C^\infty(\Omega,[0,1])$ satisfy
  \begin{displaymath}
    \eta(x)=
    \begin{cases}
      1 & x\in \Omega\setminus \Dg_\delta(\epsilon)\,, \\
      0 & x\in \Dg_{\delta/2}(\epsilon) \,.
    \end{cases}
  \end{displaymath}
  With the aid of \eqref{eq:117} and \eqref{eq:112}, it can be easily
  verified, for some $C=C(p,J,\Omega)>0$ which is independent of both $\delta$
  and $\epsilon$, that for all $0<\delta<\delta_0$ we can construct $\eta$ with the
  additional property
  \begin{displaymath}
   |\nabla\eta| \leq \frac{C}{\delta\epsilon} \,.
  \end{displaymath}
Let further
$$\zeta=\chi\eta\,,
$$
where
\begin{displaymath}
  \chi=
  \begin{cases}
    \exp\big(\alpha_\delta\epsilon^{-1-\alpha/2}d(x,\Dg_\delta(\epsilon))\big) &\text{if } x\in\Omega\setminus \Dg_\delta(\epsilon)\,, \\
    1 &\text{if } x\in \Dg_\delta(\epsilon)\,.
  \end{cases}
\end{displaymath}
We leave the determination of
$\alpha_\delta$ to a later stage.

Multiplying (\ref{eq:25}a) by $\zeta^2\bar{\psi_\epsilon}$ and integrating by parts yields
\begin{displaymath}
  \|\nabla_{\epsilon^{-2}A_\epsilon}(\zeta\psi_\epsilon)\|_2^2 =\frac{1}{\epsilon^2}\big[ \|\zeta\psi_\epsilon\|_2^2 -
  \|\zeta^{1/2}\psi_\epsilon\|_4^4\big] + \|\psi_\epsilon\nabla\zeta\|_2^2 \,.
\end{displaymath}
By \eqref{eq:22} we have that for sufficiently small $\epsilon$
\begin{displaymath}
 \|\nabla_{\epsilon^{-2}A_\epsilon}(\zeta\psi_\epsilon)\|_2^2  \geq\frac{\Theta_0\delta}{2\epsilon^2}\epsilon^{-\alpha}
 \|\zeta\psi_\epsilon\|_2^2 \,,
\end{displaymath}
where $\Theta_0$ is defined in \eqref{eq:20}.  Consequently,
\begin{equation}
\label{eq:122}
   \frac{\Theta_0\delta \epsilon^{-\alpha}-2}{2\epsilon^2}\|\zeta\psi_\epsilon\|_2^2 \leq \|\psi_\epsilon\nabla\zeta\|_2^2 \,.
\end{equation}
From \eqref{eq:122} we learn that
\begin{displaymath}
  \|\zeta\psi_\epsilon\|_2 \leq \Big[\frac{2\epsilon^2}{\Theta_0\delta \epsilon^{-\alpha}-2}\Big]^{1/2} \Big(\alpha_\delta
  \epsilon^{-(1+\alpha/2)}\|\zeta\psi_\epsilon\|_2 +
  \|\psi_\epsilon\nabla\eta\|_2\Big) \,,
\end{displaymath}
where we have used the fact that $|\nabla d(x,\Dg_\delta(\epsilon))|\leq1$  a.e.

Choosing
\begin{displaymath}
  \alpha_\delta=\frac{1}{2}\Big[\frac{\Theta_0\delta}{2}\Big]^{1/2}
\end{displaymath}
we obtain that for sufficiently small $\epsilon$
\begin{displaymath}
  \|\zeta\psi_\epsilon\|_2 \leq 4\Big[\frac{2}{\delta}\Big]^{1/2} \epsilon^{1+\alpha/2}\|\psi_\epsilon\nabla\eta\|_2 \leq
  \frac{C\epsilon^{\alpha/2}}{\delta^{3/2}}\|\psi_\epsilon\|_{L^2(\Dg_\delta(\epsilon)\setminus \Dg_{\delta/2}(\epsilon))}\,.
\end{displaymath}
The proof of \eqref{eq:121} now follows from \eqref{eq:106}.
\end{proof}

We now obtain an improved lower bound.
\begin{proposition}
  Let $\delta_0$ be given by \eqref{eq:118}.  Under the conditions of Lemma
  \ref{Lem3}, for any $0<\delta<\delta_0$,  there exist positive constants  $C_\delta$ and $\epsilon_\delta$ such that for
   $0< \epsilon \leq \epsilon_\delta$ we have
  \begin{equation}
\label{eq:123}
d_\delta(\epsilon) \geq  C_\delta \,.
  \end{equation}
\end{proposition}
\begin{proof}
  We begin by noticing that by (\ref{eq:104}b), \eqref{eq:113}, and
  \eqref{eq:121} we have
\begin{displaymath}
  \|\nabla B_\epsilon\|_{L^2(\Omega\setminus\Dg_\delta(\epsilon))} \leq
  \|\nabla\phi_\epsilon\|_{L^2(\Omega\setminus\Dg_\delta(\epsilon))}+Ce^{-1/\epsilon} \leq \widehat C \epsilon^{-\alpha} \,.
\end{displaymath}
Then, we write
\begin{displaymath}
  (\delta_0-\delta)\epsilon^{-\alpha} \leq \|\nabla B_\epsilon\|_{L^1(\Omega\setminus\Dg_\delta(\epsilon))} \leq Cd_\delta^{1/2}
  \|\nabla B_\epsilon\|_{L^2(\Omega\setminus\Dg_\delta(\epsilon))} \leq Cd_\delta^{1/2}\epsilon^{-\alpha}\,,
\end{displaymath}
from which \eqref{eq:123} readily follows.
\end{proof}
\begin{remark}
\label{rem:size}
  Note that by \eqref{eq:121} and  the above arguments, \eqref{eq:123}
  holds even if we use $d_\delta$ to measure the distance of $\Dg_\delta(\epsilon)$ from
  any point on $\partial\Omega$ where $B_\epsilon>\delta_1\epsilon^{-\alpha}$ with $\delta_1>\delta$.
\end{remark}

Proposition \ref{semiclassical} now follows from \eqref{eq:121} and \eqref{eq:123}.

\paragraph{\bf Acknowledgements}
The authors gratefully acknowledge Phuc Nguyen for introducing to them
Byun's result \cite{by07} and Monique Dauge for the idea behind the
proof of Lemma \ref{lem:Dirichlet-Neumann}. Y. Almog was partially
supported by NSF grant DMS-1109030 and by US-Israel BSF grant
no.~2010194.  B. Helffer was partially supported by the ANR programme
NOSEVOL and Simons foundation visiting fellow at the Newton Institute
in Cambridge during the completion of this article. X.B. Pan was
partially supported by the National Natural Science Foundation of
China grant no. 11171111 and no. 11431005. 
\vspace{2ex}

The authors declare that
they have no conflict of interest.


\end{document}